\documentclass[a4paper,UKenglish,cleveref, autoref, thm-restate,dvipsnames]{lipics-v2021}

\usepackage{todonotes}
\usepackage{xspace}
\usepackage{mathtools}
\usepackage{multicol}
\usepackage{stmaryrd}
\usepackage{amsthm}
\usepackage{amssymb}
\usepackage{nicefrac}
\usepackage{xparse}
\usepackage{array}
\usepackage{makecell}

\usepackage[appendix=inline]{apxproof} 
\newtheoremrep{apxtheorem}[theorem]{Theorem}
\newtheoremrep{apxlemma}[theorem]{Lemma}

\nolinenumbers

\title{One-clock synthesis problems} 

\author{S{\l}awomir Lasota}{University of Warsaw}{}{https://orcid.org/0000-0001-8674-4470}{Partially supported by ERC grant INFSYS, agreement no. 950398 and 
by the NCN grant 2024/55/B/ST6/01674.}
\author{Mathieu Lehaut}{University of Warsaw}{}{https://orcid.org/0000-0002-6205-0682}{Partially supported by the NCN grant 2021/41/B/ST6/00535.}
\author{Julie Parreaux}{Univ Rennes, Inria, CNRS, IRISA}{}{https://orcid.org/0009-0009-2744-780X}{Partially supported by the ANR project BisoUS (ANR-22-CE48-0012).}
\author{Rados{\l}aw Pi{\'o}rkowski}{University of Oxford}{}{https://orcid.org/0000-0002-9643-182X}{}

\authorrunning{S.~Lasota, M.~Lehaut, J.~Parreaux, R.~Pi{\'o}rkowski} 

\Copyright{S{\l}awomir Lasota, Mathieu Lehaut, Julie Parreaux, Rados{\l}aw Pi{\'o}rkowski} 



\category{} 

\relatedversion{} 

\EventEditors{John Q. Open and Joan R. Access}
\EventNoEds{2}
\EventLongTitle{42nd Conference on Very Important Topics (CVIT 2016)}
\EventShortTitle{CVIT 2016}
\EventAcronym{CVIT}
\EventYear{2016}
\EventDate{December 24--27, 2016}
\EventLocation{Little Whinging, United Kingdom}
\EventLogo{}
\SeriesVolume{42}
\ArticleNo{23}

\newcommand{\pspace}{\textsc{PSpace}\xspace}

\newcommand{\para}[1]{\vspace{-2mm}\subparagraph*{#1.}}
\newcommand{\Buchi}{B\"uchi\xspace}

\newcommand{\wlogg}{w.l.o.g.\xspace}
\newcommand{\tuple}[1]{\ensuremath{\langle #1 \rangle}\xspace}

\newcommand{\R}{\mathbb{R}}

\newcommand{\N}{\mathbb{N}}
\newcommand{\Npos}{\mathbb{N}_{>0}}
\newcommand{\Q}{\mathbb{Q}}
\newcommand{\Time}{\mathbb{T}}
\newcommand{\Qpos}{\Q_{\geq 0}}
\newcommand{\Rpos}{\R_{\geq 0}}

\newcommand{\rresNTA}{\ensuremath{\text{reach-}\textsc{NTA}^{\text{res}}\xspace}}
\newcommand{\resNTA}{\ensuremath{\textsc{NTA}^{\text{res}}}}
\newcommand{\rNTA}{\ensuremath{\text{reach-}\textsc{NTA}}\xspace}
\newcommand{\NTA}{\ensuremath{\textsc{NTA}}\xspace}
\newcommand{\NRA}{\ensuremath{\textsc{NRA}}\xspace}
\newcommand{\DTA}{\ensuremath{\textsc{DTA}}\xspace}
\newcommand{\DRA}{\ensuremath{\textsc{DRA}}\xspace}
\newcommand{\DFA}{\ensuremath{\textsc{DFA}}\xspace}

\newcommand{\A}{\mathcal{A}}

\newcommand{\C}{\mathcal{C}}

\newcommand{\Cl}{\mathcal{X}}
\newcommand{\val}{\nu}
\newcommand{\Gu}{\mathrm{Guard}}
\newcommand{\wordToVal}{\mathrm{val}}

\newcommand{\Locs}{L}
\newcommand{\loc}{\ell}
\newcommand{\LocsI}{\Locs_i}
\newcommand{\LocsF}{\Locs_f}

\newcommand{\Trans}{\Delta}

\newcommand{\NTAEx}{\ensuremath{\tuple{\Locs, \Cl, \Sigma, \LocsI, \LocsF, \Trans}}}

\newcommand{\Lang}{\mathcal{L}}

\newcommand{\zeroval}{{0}^\Counter}
\newcommand{\M}{\mathcal{M}}
\newcommand{\inc}{\mathrm{inc}}
\newcommand{\dec}{\mathrm{dec}}
\newcommand{\zt}{\mathrm{zt}}
\newcommand{\ztest}{\mkern3mu{\sequals\mkern-1mu{0}{?}}}
\newcommand{\increment}{\mkern2mu{\splus\mkern-1mu\splus}\mkern1mu}
\newcommand{\decrement}{\mkern2mu{\sminus\mkern-1mu\sminus}\mkern1mu}
\newcommand{\op}{\mathrm{op}}
\newcommand{\Counter}{\mathcal{C}}
\newcommand{\counter}{c}
\newcommand{\States}{S}
\newcommand{\state}{s}
\newcommand{\Instructions}{\mathcal{I}}
\newcommand{\instruction}{I}

\newcommand{\MRuns}{\mathrm{Runs}}
\newcommand{\run}{\rho}

\renewcommand{\Game}{\mathcal{G}}
\newcommand{\Subject}{\ensuremath{\mathrm{Agent}}\xspace}
\newcommand{\Owner}{\ensuremath{\mathrm{Owner}}\xspace}
\newcommand{\subject}{agent\xspace}
\newcommand{\owner}{owner\xspace}
\newcommand{\Timer}{\ensuremath{\mathrm{Timer}}\xspace}
\newcommand{\Monitor}{\ensuremath{\mathrm{Monitor}}\xspace}

\newcommand{\TimerAlphabet}{\mathrm{T}}
\newcommand{\MonitorAlphabet}{\mathrm{M}}

\newcommand{\Controller}{\A}
\newcommand{\strategy}{\sigma}
\newcommand{\strategyT}{\strategy_\mathrm{T}}
\newcommand{\strategyM}{\strategy_\mathrm{M}}

\newcommand{\Win}{W_\M} 
\newcommand{\Reach}{\mathrm{Reach}}
\newcommand{\ok}{\ensuremath{\cmark}\xspace}
\newcommand{\err}{\ensuremath{\xmark}\xspace}

\newcommand{\Ok}{\ensuremath{\cmark}\xspace}
\newcommand{\errR}{\ensuremath{\xmark}_\mathrm{R}\xspace}
\newcommand{\errA}{\ensuremath{\xmark}_\mathrm{A}\xspace}
\newcommand{\errB}{\ensuremath{\xmark}_\mathrm{B}\xspace}
\newcommand{\errC}{\ensuremath{\xmark}_\mathrm{C}\xspace}

\newcommand{\RegI}{\mathrm{Reg}}
\newcommand{\RegIM}{\RegI_{\M}}
\newcommand{\RegIMT}{\RegI_{\M}^\Time}
\newcommand{\RegIMTloc}{\RegI_{\M}^{\Time,\ell}}

\newcommand{\RegII}{\mathrm{Reg}_{\M}}

\newcommand{\AM}{A_{\M}^{\Time}}
\newcommand{\BM}{B_{\M}^{\Time}}
\newcommand{\CM}{C_{\M}^{\Time}}

\newcommand{\AMloc}{A_{\M}^{\Time,\ell}}
\newcommand{\BMloc}{B_{\M}^{\Time,\ell}}
\newcommand{\CMloc}{C_{\M}^{\Time,\ell}}

\newcommand{\ErrLang}{\mathrm{Err}_\M}
\newcommand{\NotErrLang}{\mathrm{Ok}_\M}
\newcommand{\ErrLangOne}{\mathrm{Err}^{\mathrm{R1}}_\M}
\newcommand{\NotErrLangOne}{\mathrm{Ok}^{\mathrm{R1}}_\M}
\newcommand{\ErrLangTwo}{\mathrm{Err}^{\mathrm{R2}}_\M}
\newcommand{\NotErrLangTwo}{\mathrm{Ok}^{\mathrm{R2}}_\M}

\newcommand{\pool}{\ensuremath{\square}}
\newcommand{\FracPart}{\mathbb{F}}
\newcommand{\fracpart}[1]{\mathrm{frac}(#1)}

\usetikzlibrary{automata,positioning,shadows}
\usetikzlibrary{calc, shapes, backgrounds,arrows}
\usetikzlibrary{chains, decorations.pathmorphing}
\usetikzlibrary{positioning,decorations.pathreplacing, automata}
\usetikzlibrary{shapes.misc}

\tikzset{every loop/.style={looseness=7}, >=latex}
\tikzset{every picture/.style={>=latex}}
\tikzstyle{PlayerMin}=[draw,circle,minimum size=4mm,inner sep=1.5pt]
\tikzstyle{PlayerMax}=[draw,rectangle,minimum size=7mm,inner sep=1.5pt]
\tikzstyle{Player}=[draw,diamond,minimum size=7mm,inner sep=1.5pt]
\tikzstyle{target}=[circle, minimum size=1mm,inner sep=-2pt]
\tikzstyle{PlayerMinmin}=[draw,circle, minimum size=1.5mm,inner sep=0pt]
\tikzstyle{PlayerMaxmin}=[draw,rectangle,minimum size=1.5mm,inner sep=0pt]
\tikzstyle{PlayerMinsmall}=[draw,circle, minimum size=3mm,inner sep=0pt]
\tikzstyle{PlayerMaxsmall}=[draw,rectangle,minimum size=3mm,inner sep=0pt]
\tikzstyle{leaf}=[draw,diamond,minimum size=7mm,inner sep=1.5pt]

\tikzstyle{strat} =[minimum width=0.1cm,line width=0.01mm,draw=none]
\tikzstyle{vecArrow} = [decoration={markings,mark=at position
1 with {\arrow[scale=.5,>=latex]{>}}}
]

\newcommand{\winContMonitor}{\mathrm{winCont}^{\mathrm{M}}_\A}
\newcommand{\intsize}{\mathrm{size}}

\newcommand{\ValEnc}{\mathrm{ValEnc}}
\newcommand{\RunEnc}{\mathrm{RunEnc}}
\newcommand{\enc}{\mathrm{enc}}
\newcommand{\sem}[1]{\llbracket#1\rrbracket}
\newcommand{\untimed}{\mathrm{untime}}
\newcommand{\projA}{\mathrm{proj}_{\TimerAlphabet,\Time}}
\newcommand{\projAA}{\mathrm{proj}_\TimerAlphabet}
\newcommand{\projB}{\mathrm{proj}_\MonitorAlphabet}

\theoremstyle{definition}

\newcommand{\cmark}{\mathord{\raisebox{-0.18pt}{\includegraphics[height=1.04\fontcharht\font`A]{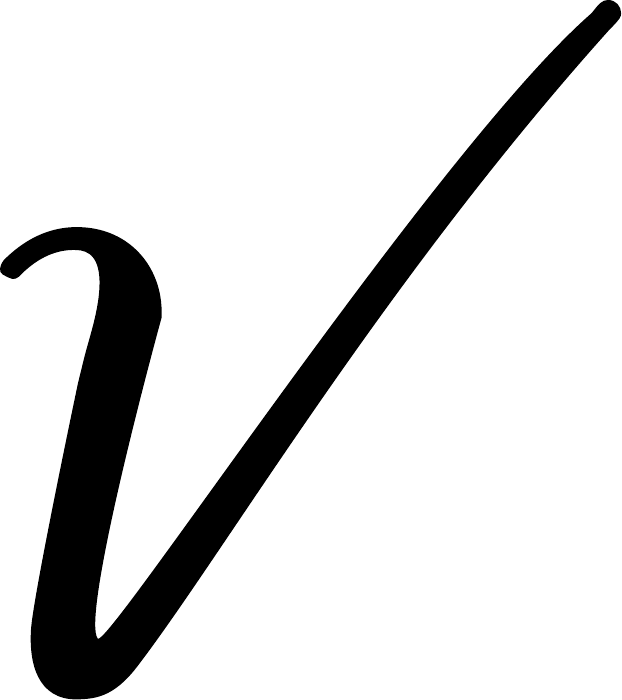}}}}
\newcommand{\xmark}{\mathord{\raisebox{-0.16pt}{\includegraphics[height=1.04\fontcharht\font`A]{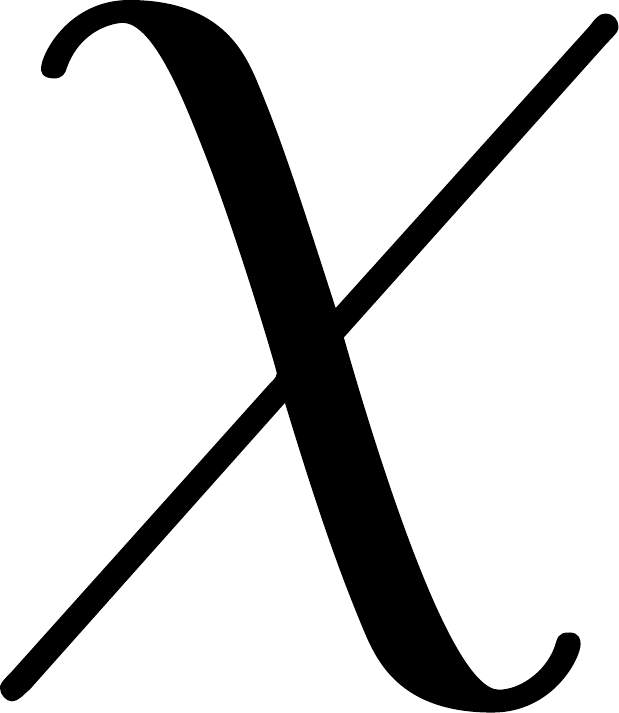}}}}
\newcommand{\sminus}{\mathop{\raisebox{-0.16pt}{\includegraphics[height=1.04\fontcharht\font`A]{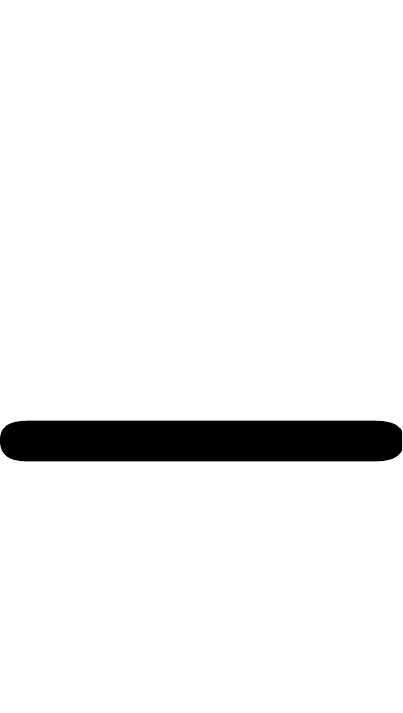}}}}
\newcommand{\splus}{\mathop{\raisebox{-0.16pt}{\includegraphics[height=1.04\fontcharht\font`A]{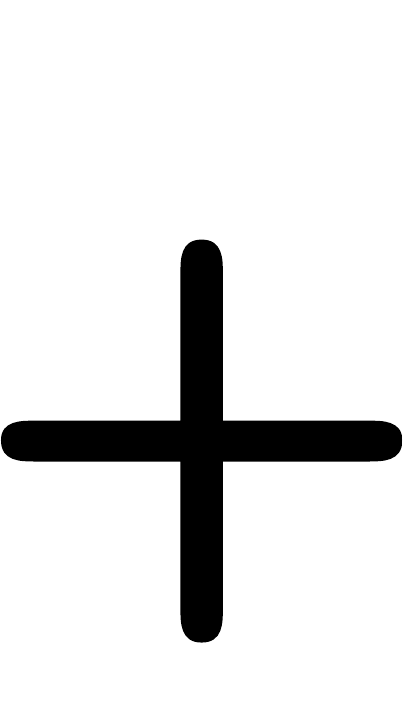}}}}
\newcommand{\sequals}{\mathop{\raisebox{-0.16pt}{\includegraphics[height=1.04\fontcharht\font`A]{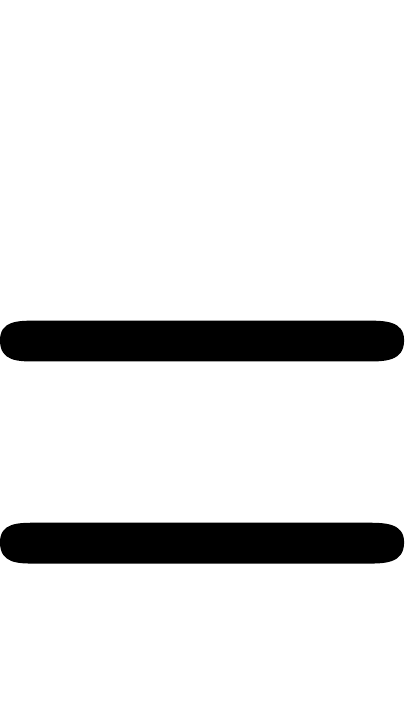}}}}

\newcommand{\suchthat}{\suchthatSymbol\PackageWarning{Radeks Macro}{Command suchthat used outside of matching PairedDelimiter was used on input line \the\inputlineno.}}
\newcommand\suchthatSymbol[1][]{\nonscript\:#1\vert\allowbreak\nonscript\:\mathopen{}}
\DeclarePairedDelimiterX{\setInner}[1]\{\}{\renewcommand\suchthat{\suchthatSymbol[\delimsize]}#1}
\NewDocumentCommand{\set}{ O{a} m }{\ifx#1a\setInner*{#2}\else\ifx#1b{\{#2\}}\else\setInner[#1]{#2}\fi\fi}

\newcolumntype{C}[1]{>{\centering\arraybackslash}p{#1}}
\newcolumntype{M}[1]{>{\centering\arraybackslash}m{#1}}

\definecolor{my-purple-the-lightest}{HTML}{fbf7ff}
\definecolor{my-purple-very-light}{HTML}{f1ebf7}
\definecolor{my-purple-very-light-half}{HTML}{6937A5}
\tikzset{myPurpleVeryLightOverlay/.style={fill=my-purple-very-light-half,fill opacity=0.1,text opacity=1}}
\definecolor{my-purple-light}{HTML}{d5c0ea}
\definecolor{my-purple}{HTML}{6700CE}
\definecolor{my-purple-dark}{HTML}{3C0078}
\definecolor{my-purple-very-dark}{HTML}{1c0038}

\definecolor{my-grey-the-lightest}{HTML}{FAFAFA}
\definecolor{my-grey-very-light}{HTML}{f2f2f2}
\definecolor{my-grey-very-light-half}{HTML}{7D7D7D}
\tikzset{myGreyVeryLightOverlay/.style={fill=my-grey-very-light-half,fill opacity=0.1,text opacity=1}}
\definecolor{my-grey-light}{HTML}{E3E3E3}
\definecolor{my-grey}{HTML}{D4D4D4}
\definecolor{my-grey-dark}{HTML}{777777}
\definecolor{my-grey-very-dark}{HTML}{404040}
\definecolor{my-grey-the-darkest}{HTML}{111111}

\definecolor{my-yellow-the-lightest}{HTML}{fff6d9}
\definecolor{my-yellow-very-light}{HTML}{ffedb3}
\definecolor{my-yellow-light}{HTML}{ffdf78}
\definecolor{my-yellow}{HTML}{fcc512}

\definecolor{c1}{HTML}{00C5DF}
\definecolor{c2}{HTML}{FF2EB4}
\definecolor{c3}{HTML}{86CF00}
\definecolor{c4}{HTML}{FF6B00}
\definecolor{c5}{HTML}{AD30FF}

\usetikzlibrary{patterns,arrows,arrows.meta,calc,positioning,tikzmark,decorations.pathmorphing,decorations.pathreplacing}

\tikzstyle{overbrace text style}=[font=\small, above, pos=.5, yshift=2mm]
\tikzstyle{overbrace style}=[decorate,decoration={brace,raise=1mm,amplitude=3pt}]
\tikzstyle{underbrace style}=[decorate,decoration={brace,raise=1mm,amplitude=3pt,mirror}]
\tikzstyle{underbrace text style}=[font=\small, below, pos=.5, yshift=-2mm]

\begin{document}

\maketitle

\begin{abstract}
We study a generalisation of Büchi-Landweber games to the timed setting. 
The winning condition is specified by a non-deterministic timed automaton,
and one of the players can elapse time. We perform a systematic study of synthesis problems in 
all variants of timed games, depending on which player's winning condition 
is specified, and which player's strategy (or controller, a finite-memory strategy) 
is sought. 

As our main result we prove ubiquitous undecidability in all the variants, 
both for strategy and controller synthesis, already for winning 
conditions specified by one-clock automata. This strengthens and generalises 
previously known undecidability results. We also fully characterise those cases 
where finite memory is sufficient to win, namely existence of a strategy
implies existence of a controller.

All our results are stated in the \emph{timed} setting, while analogous results 
hold in the \emph{data} setting where one-clock automata are replaced by
one-register ones. 
\end{abstract}

\newpage

\section{Introduction}
\label{sec:intro}


Nondeterministic timed automata (\NTA)
are one of the most widespread models of real-time reactive systems with a huge literature.
They consist of finite automata extended with real-valued clocks
which can be reset and compared by inequality constraints.
Since the  seminal paper showing  \pspace-completeness of the reachability problem~\cite{Alur1994},
%
\NTA found their  way to the automatic verification of timed systems,
eventually leading to mature tools such as UPPAAL~\cite{Behrmann:2006:UPPAAL4},
UPPAAL Tiga (timed games)~\cite{CassezDavidFleuryLarsenLime:CONCUR:2005},
and PRISM (probabilistic timed automata)~\cite{KwiatkowskaNormanParker:CAV:2011}.
%
%
One of the restrictions of this model is the lack of determinisation and closure under complementation
of \NTA languages, as well as undecidability of universality/inclusion problems~\cite{Alur1994}.
One recovers decidability for a restricted subclass of \NTA with one clock ($\NTA_1$)~\cite{OW04},
and even for alternating timed automata with one clock~\cite{Lasota2008}.

\emph{Deterministic timed automata} (\DTA) form a strict subclass of \NTA
where the next configuration is uniquely determined by the current one and the timed input symbol.
This class enjoys stronger properties,
such as decidable inclusion problems and complementability~\cite{Alur1994},
and it is used in several applications,
such as test generation~\cite{NielsenSkou:STTT:2003},
fault diagnosis~\cite{BouyerChevalierDSouza:FOSSACS:2005},
learning~\cite{VerwerWeerdtWitteveen:Benelearn:2007,TapplerAichernigLarsenLorber:FOMATS:2019};
timed games~\cite{AsarinMaler:HSCC:1999,JurdzinskiTrivedi:ICALP:2007,BrihayeHenzingerPrabhuRaskin:ICALP:2007},
and recognisability of timed languages~\cite{Maler:Pnueli:FOSSACS:04}.

\para{Timed games and synthesis problems}
There are many variants of timed games in the literature,
depending on whether the players must enforce a non-Zeno play,
who controls the elapse of time,
concurrent actions, winning conditions, etc.~\cite{AsarinMaler:HSCC:1999,JurdzinskiTrivedi:ICALP:2007,BrihayeHenzingerPrabhuRaskin:ICALP:2007,Wong-Toi-Hoffmann:CDC:1991,
MalelPnueliSifakis:STACS:1995,
AsarianMalerPnueliSifakis:SSSC:1998,
DsouzaMadhusudan:STACS:2002,
deAlfaroFaellaHenzingerMajumdarStoelinga:CONCUR:2003,RF25,
jenkins2011church}.
Following prior works~\cite{Clemente0P-20,Piorkowski-PhD_22},
we consider asymmetric (only one player can elapse time), infinite-duration
turn-based games that can be considered as
timed generalisation of \Buchi-Landweber games~\cite{BuchiLandweber:AMS:1969}.
There are two players, called \Timer and \Monitor,
who play taking turns in a strictly alternating fashion.
In the $i$th round, \Timer chooses a letter~$a_i$ from a finite alphabet along with a nonnegative
rational time delay~$\tau_i$,
and \Monitor responds with a letter~$b_i$ from a finite alphabet.
Together, the players generate an infinite play
$\pi = (a_1, b_1, \tau_1) \, (a_2, b_2, \tau_2) \cdots$.
The winning set of either \Timer or \Monitor is specified
by a \emph{nondeterministic} timed automaton. For comparison,
the easier special case where the winning set is given by a \emph{deterministic} or \emph{history deterministic}
timed automaton has been previously studied (e.g.,~\cite{DsouzaMadhusudan:STACS:2002,BHLST24}).
Given our undecidability results, we do not consider a more general symmetric variant
in which both players may elapse time.

We investigate the \emph{timed reactive synthesis} problem,
which asks if a given player has a strategy ensuring that every play that conforms to the
strategy is winning for that player.
We distinguish four variants of this problem, depending on which player's
winning set is specified (note the lack of complement closure of \NTA languages)
and for which player the winning strategy is sought.
We also study the \emph{timed Church synthesis} problem, which
asks if a given player has a winning finite-memory strategy (controller).
For \Monitor, a controller is a \DTA whose
transitions output letters from \Monitor's alphabet;
for \Timer, it is a \DFA whose
transitions output letters from \Timer's alphabet along with time delays.

This study of timed generalisations of \Buchi--Landweber games was initiated in~\cite{Clemente0P-20} in the setting where \Timer's winning set is specified by an \NTA, and the goal is to synthesise a controller for \Monitor.
The main result of that work is the decidability of the \emph{resource-bounded} timed Church synthesis problem, in which one seeks a controller using at most~$k$ clocks, for a fixed~$k$.
Subsequently, \cite{Piorkowski-PhD_22} established the undecidability of the unrestricted (resource-unbounded) version, even when \Timer's winning set is given by a two-clock \NTA and \Monitor's controller is sought.
The decidability of this problem for $\NTA_1$ (one-clock \NTA) winning sets remained open, as did the status of other variants and of timed reactive synthesis.
The present paper resolves all of these open questions in the negative.

\para{Contribution}
Given as many as eight different decision problems, one could expect
a complex decidability/complexity landscape. 
As our main technical contribution, we show that this is not the case: all
eight decision problems are undecidable already in the simplest case
where the winning sets are specified by $\NTA_1$.
These undecidability results significantly strengthen and generalise 
previously known lower bounds.
They also demonstrate that restricting to $\NTA_1$ does not lead to recovery
of decidability of synthesis (game solving) problems, 
similarly as restricting to $\NTA_1$ does not yield decidability of language
universality and related problems over infinite words
(as opposed to universality and related problems for $\NTA_1$
languages of finite timed words~\cite{OW04}).

Proving undecidability for eight problems required four reductions.
The cases where \Monitor's winning set is specified by an $\NTA_1$ are easily shown
undecidable by reductions
from the $\NTA_1$ universality and sampled universality problems~\cite{Abdulla2007,Lasota2008}.
Undecidability proofs in the other case, where \Timer's winning set is specified, 
constitute the technical core of the paper, and proceed by reductions from two
undecidable problems for lossy counter machines: boundedness and
repeated reachability~\cite{Mayr03,Schnoebelen10a}.
 
Finally, we also address the question of when finite memory is sufficient to win, 
that is when existence of a strategy implies existence of a controller.
On one hand, 
we prove this finite-memory property for the player whose winning set is specified:
for all $\NTA_1$ specifications if this player is \Monitor, and only for the restricted
case of \emph{reachability} $\NTA_1$ specifications if this player is \Timer.
On the other hand, we  demonstrate that the finite-memory property fails
in all other cases: it is not satisfied in general by \Timer when his winning set is specified;
and the property fails for both \Timer and \Monitor when the opponent's winning set is specified,
even in case of reachability $\NTA_1$ specifications.

It is well known 
that \NTA exhibit close similarity to \emph{nondeterministic register automata}
(\NRA), known also as finite-memory automata~\cite{KF94}
(see e.g.~\cite{FHL16} for a formal connection between the models).
Register automata input data values (in place of timestamps) 
and use registers (in place of clocks) to store data values
(\cite{Segoufin06} is an excellent survey of automata models for data setting).
For many language-related problems, like emptiness (reachability), universality
or inclusion, register automata admit exactly the same  
(un)decidability results~\cite{DL09}
that are known for timed automata~\cite{OW04,Lasota2008}.
Language closure properties are also analogous in both settings.
Confirming these deep similarities further,~\cite{Piorkowski-PhD_22} proves  undecidability of
the \emph{timed Church synthesis} problem for winning sets specified by \NTA with two clocks, 
as well as undecidability of the \emph{data Church synthesis} 
for winning sets specified by \NRA with two registers.
Similar studies of data generalisation of \Buchi-Landweber games were also conducted 
independently in~\cite{LeoThesis} and works cited therein~\cite{EFR19,EFK22}.
All our undecidability results transfer from the timed to the data setting:
undecidability holds for all the eight variants of 
\emph{data reactive/Church synthesis} problems,
already when winning sets are specified by $\NRA_1$ (\NRA with one register).
(These further results exceed the scope of the present paper and are planned to
be included in the forthcoming full version of the paper.)

\para{Outline}

We start by defining the setting of timed games and the synthesis problems
(\cref{sec:prelim}).
We also discuss the finite-memory property there.
\Cref{sec:universality} is a warm-up where we deal with the easy cases of synthesis problems
(\Monitor's winning set is specified).
Then in \cref{sec:lossy} we define the undecidable problems for lossy counter machines, 
to be used in the main technical part of the paper,
\cref{sec:Buchi,sec:boundedness}.
In the two latter sections we provide the undecidability proofs in
the hard case where \Timer's winning set is specified.
The last section contains final remarks.

%

\section{Timed synthesis problems}
\label{sec:prelim}

Let $\Qpos$ and $\Npos$ be nonnegative rational numbers, and positive integers, respectively.
We let $\Cl$ be a finite set of variables called \emph{clocks}.
A \emph{valuation} is a mapping $\val \colon \Cl \to \Qpos$
(we prefer to use rational values instead of real ones).
For a valuation $\nu$, a delay $\tau \in \Qpos$ and a subset $Y \subseteq \Cl$ of clocks,
we define the valuation $\val + \tau$ as $(\val + \tau)(x) = \val(x) + \tau$, for all $x \in \Cl$,
and the valuation $\val[Y \coloneqq 0]$ as $(\val[Y \coloneqq 0])(x) = 0$
if $x \in Y$ , and $(\val[Y \coloneqq 0])(x) = \val(x)$ otherwise.
A \emph{guard} on clocks $\Cl$ is a
conjunction of atomic constraints of the form $x \bowtie c$, where 
$x\in\Cl$, $\bowtie\ \in \{\leq, <, =, >, \geq\}$
and $c \in \N$.
An empty guard is denoted by $\top$.
A valuation $\val$ \emph{satisfies} an atomic constraint $x \bowtie c$ if $\val(x) \bowtie c$.
The satisfaction relation is extended to all guards $g$ naturally, and denoted by $\val \models g$.
We let $\Gu(\Cl)$ denote the set of guards over $\Cl$.

Let $\Sigma$ be a finite alphabet.
By a \emph{timed word} over $\Sigma$ we mean a finite or infinite sequence 
$w = (a_1, t_1) \, (a_2, t_2)  \cdots \in (\Sigma\times\Qpos)^\omega$
where the \emph{timestamps} $t_i$ are monotone:
$t_1 \leq t_{2} \leq t_3 \leq \cdots$.
Pairs $(a, t)\in \Sigma\times\Qpos$ are sometimes called \emph{timed letters}.
Unless stated otherwise, we work with infinite timed words.
The sequence of timestamps is determined uniquely by the sequence of
\emph{delays} $\tau_1, \tau_2, \ldots \in \Qpos$, where $\tau_i = t_{i} - t_{i-1}$ (assuming $t_0 = 0$).
The \emph{untiming} of $w$,
denoted by $\untimed(w)$, is the word $a_1 \, a_2 \ldots\in \Sigma^\omega$ 
obtained by removing timestamps.
A timed language is any set of timed words over a fixed alphabet $\Sigma$.

\para{Timed automata}

    A \emph{non-deterministic timed automaton} ($\NTA$)
    $\A=\NTAEx$ consists of
    $\Locs$, a finite set of \emph{locations} with $\LocsI, \LocsF \subseteq \Locs$
    denoting the sets of initial and accepting locations, respectively;
    $\Cl$, a finite set of clocks;
    $\Sigma$, a finite alphabet; and
    $\Trans \subseteq \Locs \times \Sigma \times \Gu(\Cl) \times 2^\Cl \times \Locs$,
    a finite set of transitions.
    A transition $(\ell, a, g, Y, \ell') \in \Trans$ is written as $\ell \xrightarrow{a, g, Y} \ell'$, omitting $g$ and $Y$ whenever they are $\top$ or $\emptyset$, respectively.
    A set $A$ in place of $a$ specifies a set of transitions.
The \emph{semantics} of $\A \in \NTA$ is a timed transition system
$\sem{\A} = (Q, Q_I, \to)$ such that:
$Q = \Locs \times (\Qpos)^{\Cl}$ is the set of \emph{configurations} (i.e.,
location-valuation pairs), $Q_I = \LocsI \times \{0^\Cl\}$ is the set of initial configurations, and
$\to$ is the set of \emph{edges}.
We have $(\loc, \val) \xrightarrow{\delta, \tau} (\loc', \val')$ 
if and only if $\delta = (\loc, a, g, Y, \loc')\in \Trans$ is a transition of $\A$
such that $\val + \tau \models g$, and $\val' = (\val + \tau)[Y \coloneqq 0]$.
A \emph{run} $\run$ in $\A$ is a sequence of edges 
$\run = (\loc_1, \val_1) \xrightarrow{\delta_1, \tau_1} (\loc_2, \val_2) \xrightarrow{\delta_2, \tau_2}
\ldots$ 
of $\sem{\A}$ such that $\loc_1 \in \LocsI$.
A timed word $w = (a_1, t_1) \, (a_1, t_2) \ldots$ over $\Sigma$ \emph{labels}
a run $\run$ of $\A$ when, for all $i \in \Npos$, $\delta_i = (\loc_i, a_i, g, Y, \loc_{i+1})$
and $t_i = \sum_{j= 1}^i \tau_j$.
We adopt \emph{\Buchi} acceptance: a run is \emph{accepting} if is visits an accepting location infinitely often,
in which case we say that the word that labels that run is \emph{accepted} by $\A$.
The \emph{language} of $\A \in \NTA$,
denoted  $\Lang(\A)$, is the set of timed words accepted by $\A$.

We also use  \emph{reachability} acceptance, where a run is accepting
once it visits an accepting location at least once (like in the setting of finite words).
With this acceptance, the language $L$ of \emph{infinite} timed words accepted by an \NTA 
is determined uniquely by the language $L'$ of \emph{finite} timed words having a run ending in an accepting location. We write $L = \Reach(L')$.
So defined \emph{reachability $\NTA$} can be seen as a (strict) subclass of $\NTA$, 
denoted as \emph{$\rNTA$}.

Apart from \rNTA, we consider also other subclasses of \NTA.
For instance $\NTA_k$ consists of nondeterministic timed automata with a fixed number
$k=|\Cl|$ of clocks.
\DTA is the class of \emph{deterministic}
timed automata, where $|\LocsI| = 1$ and for all locations $\loc\in\Locs$ and
letters $a \in \Sigma$, the guards $g$ appearing in transitions of the form $(\loc, a, g, \_, \_)$
form a partition of $\Qpos^{|\Cl|}$.
We also consider a superclass $\resNTA$ 
of \emph{1-resetting} timed automata (defined in \cref{sec:boundedness}),
an extension by a limited form of $\varepsilon$-transitions that reset a clock every time it equals 1.
The class of \NTA with $\varepsilon$-transition is strictly more expressive than \NTA~\cite{BerardPVG-98}; we discuss
this choice at the end of this section, and in \cref{sec:boundedness,sec:conc}.

We combine the notation for sub- and superclasses and write $\rNTA_1$, 
$\rresNTA_1$, for the one-clock automata from \rNTA, and for the 1-resetting extension thereof.

Given $L$ a timed language, we denote by $\widehat{L}$ its \emph{complement}.
The \emph{universality problem} asks, given $\A$, if $\widehat{\Lang(\A)} = \emptyset$. 
This problem is known to be undecidable for $\NTA_1$~\cite{Lasota2008}.

\para{Timed games}

In this paper we consider turn-based games between two players called
here \Timer and \Monitor, 
where a winning condition is given by an $\NTA_1$ (resp.~$\resNTA_1$).

\begin{definition}
A \emph{timed game} $\Game=\tuple{\TimerAlphabet, \MonitorAlphabet, \A, \Owner, \Subject}$ consists of
finite alphabets $\TimerAlphabet, \MonitorAlphabet$ of $\Timer$ and $\Monitor$, respectively,
$\A\in\NTA_1$ (resp.~$\A\in\resNTA_1$) a timed automaton
over the product alphabet $\TimerAlphabet \times \MonitorAlphabet$ whose language $\Lang(\A)$ defines
the winning condition, 
and $\Owner, \Subject \in \{\Timer, \Monitor\}$.%
\end{definition}
Intuitively, \Owner specifies the player who wins a play if it belongs to $\Lang(\A)$.
As \NTA\ are not stable by
complement~\cite{Alur1994}, the choice of the \owner of the winning condition is not innocuous,
and the winning condition of the opponent may no longer be an \NTA language.
We investigate decision problems asking about existence of a winning strategy of
\Subject.
Again, the choice of \Subject is not innocuous, as we do not know if the games studied by us 
are determined.
In view of our undecidability results, we do not consider a symmetric variant of timed games
where both players would submit time delays, as this
variant would generalise the above one.

A timed game proceeds in rounds.
Each $i$th round ($i\in\Npos$) starts by \Timer's choice of $a_i \in \TimerAlphabet$ and a delay $\tau_i \in \Qpos$, 
followed by a response $b_i \in \MonitorAlphabet$ of \Monitor.
This results in a \emph{play} which is a timed word 
$(a_1, b_1, t_1)\, (a_2, b_2, t_2)\, \cdots$ over the alphabet $\TimerAlphabet \times \MonitorAlphabet$,
where $t_i = \sum_{j=1}^{i} \tau_j$.
\Owner wins the play if it belongs to $\Lang(\A)$, otherwise its opponent wins.

A \emph{strategy} for a player is a function that gives its next move as a function of the moves
of its opponent up to now.
Formally, a strategy for \Timer is a function $\strategyT \colon \MonitorAlphabet^* \to \TimerAlphabet \times \Qpos$
whereas a strategy for \Monitor is a function $\strategyM \colon (\TimerAlphabet \times \Qpos)^+ \to \MonitorAlphabet$.
We say that
a play $w=(a_1, b_1, t_1)\, (a_2, b_2, t_2)\, \ldots$ \emph{conforms} to
\Timer's strategy 
$\strategyT$ (resp.~\Monitor's strategy $\strategyM$) if
all (timed) letters in it conform to the strategy's output:
$(a_i, \tau_i) = \strategyT(b_1 \cdots b_{i-1})$ for all $i\in\Npos$
(resp.,~$b_i = \strategyM((a_1,t_1)\cdots (a_i,t_i))$ for all $i\in\Npos$).
A strategy is \emph{winning} for a player if and only if every play $w$ that conforms to it is winning for
that player (i.e., $w\in\Lang(\A)$ if and only if this player is the \owner).

A \emph{controller} is a finite-memory strategy represented by a \DTA with outputs.
Specifically, a \Monitor's controller is a \DTA with outputs from $\MonitorAlphabet$, whose
transitions are of the form
$\delta = (\loc, a, g, Y, \loc', b)$, where $b\in\MonitorAlphabet$ is the output.
The controller induces the strategy $\strategyM$ that maps $w=(a_1,t_1)\cdots (a_i,t_i)$ to the last output
of the run over $w$.
For \Timer, the notion of controller is more tricky since it needs to produce timestamps.
We let \Timer's controllers output pairs $(a,\tau)\in \TimerAlphabet \times \Qpos$,
where $\tau$ is interpreted as the next \emph{delay}.
Formally, it is defined as a \DFA (\DTA with no clocks) with  outputs from $\TimerAlphabet\times\Qpos$,
whose transitions are of the form
$\delta = (\loc, b,\loc', a, \tau)$, where $(a,\tau)\in\TimerAlphabet\times\Qpos$  is the output,
additionally equipped with an \emph{initial move} $(a_1, \tau_1)\in\TimerAlphabet\times\Qpos$.
Note the apparent restriction of \Timer's controllers, compared to general strategies, namely
the set of delays used by a controller is finite.
A controller induces the strategy $\strategyT$ that maps the empty word to $(a_1, \tau_1)$,
and a nonempty word $b_1 \cdots b_i$ to the last output
of the run over $w$.

\para{Timed synthesis problems}

We investigate decision problems to determine whether 
$\Subject$ wins, i.e., whether
\Subject has a winning strategy.
We distinguish four cases, 
depending on the choice of $\Owner \in \{\Timer, \Monitor\}$
and $\Subject \in \{\Timer, \Monitor\}$.
%
Furthermore, in each of the four cases we consider two distinct decision problems: given a timed game, 

\begin{itemize}
    \item  the \emph{timed reactive synthesis} asks if \Subject
    has a winning strategy;
    \item the \emph{timed Church synthesis} asks if \Subject
    has a winning controller.
\end{itemize}

\begin{table}[t]
    \centering
    \begin{tabular}{|m{2.8cm}|M{4.25cm}|M{4.25cm}|}
        \cline{2-3}
        \multicolumn{1}{c|}{} & \multicolumn{1}{c|}{$\Subject = \Timer$} & \multicolumn{1}{c|}{$\Subject = \Monitor$} \\
        \hline
        \makecell[c]{$\Owner = \Timer$} & coincide for $\A\in\rNTA_1$,\par but not for all $\A\in\NTA_1$ & do not coincide,\par even for $\A\in\rNTA_1$ \\
        \hline
        \makecell[c]{$\Owner = \Monitor$} & do not coincide,\par even for $\A\in\rNTA_1$ & coincide for all $\A\in\NTA_1$ \\
        \hline
    \end{tabular}
    \caption{Reactive vs.~Church synthesis.}
    \label{tab:coinc}
\end{table}

\noindent
The problems coincide in some cases
(see Table~\ref{tab:coinc} for a summary).
Most importantly, if \Monitor is both the \owner of the winning condition, and also the \subject of
the synthesis problem, then whenever it has a winning strategy it also has a controller.
Likewise for \Timer, but only in the restricted case of $\rNTA_1$ winning conditions:


\begin{apxtheoremrep}
    \label{thm:equivalence}
    The two synthesis problems coincide, meaning that if \Subject has a winning strategy then it has a winning controller, in the following cases:
    \begin{enumerate}
        \item $\Owner=\Subject=\Timer$ and $\A \in$ $\rNTA_1$.
        \item $\Owner=\Subject=\Monitor$ and $\A \in \NTA_1$.
    \end{enumerate}
\end{apxtheoremrep}

\begin{proof}
	We provide the proof of the two cases mentioned by this theorem.\label{apxproof:equivalence}
	In both cases, having a winning controller implies having a winning strategy, so only the other direction needs to be proved.

	\paragraph*{\boldmath Proof of the first case: $\Owner = \Subject = \Timer$}
	\label{subsec:equivalence-finite}


	\begin{lemma}
		\label{lem:equivalence-finite}
		Let $\Game = \tuple{\TimerAlphabet,\MonitorAlphabet,\A,\Timer,\Timer}$ be a game with $\A \in \rNTA_1$.
		If \Timer has a winning strategy $\strategyT$, then it has a winning controller.
	\end{lemma}
	\begin{nestedproof}
		We consider the tree of all plays of $\Game$ that conform to $\strategyT$.
		There are \Timer nodes with only one child (the unique pair in $\TimerAlphabet \times \Qpos$ output by $\strategyT$ on the play so far), and \Monitor nodes with one child per letter in $\MonitorAlphabet$.
		Since we consider infinite plays, this tree is also infinite.
		We show that it is enough to consider only a finite part of this tree to build a winning controller.

		By K\H{o}nig's lemma on trees, there exists an infinite branch in the tree.
		In particular, there exists an infinite word $w$ that conforms to $\strategyT$ and that belongs to $\Lang(\A)$ since $\strategyT$ is winning.
		As $\A \in \rNTA$, there is a finite prefix $w'$ of $w$, such that any continuation of $w'$ is also in $\Lang(\A)$.
		Also $w'$ conforms to $\strategyT$ because it is a prefix of $w$.
		Consider the new tree where we cut the branch after $w'$.
		We repeat this process on all infinite branches until only a finite tree remains.
		All branches in this tree are plays that conform to $\strategyT$ and any continuation is accepted by $\A$.
		We build a controller $\Controller$ that mimics $\strategyT$ as long as the play is still in this finite tree, and then does whatever when outside of this tree.
		Thus, we get a finite controller that is winning.
	\end{nestedproof}

	\paragraph*{\boldmath Proof of the second case: $\Owner = \Subject = \Monitor$}
	\label{subsec:equivalence-infinite}


	\begin{lemma}
		\label{prop:equivalence-infinite}
		Let $\Game = \tuple{\TimerAlphabet,\MonitorAlphabet,\A,\Monitor,\Monitor}$ be a game with $\A \in \NTA_1$.
		If \Monitor has a winning strategy $\strategyM$, then it has a winning controller.
	\end{lemma}

    The remainder of this section is devoted to the proof.
	Intuitively, we consider the tree of plays of $\Game$ that conform to $\strategyM$.
	In contrast to the previous case, this tree has unbounded branching:
	each \Timer node has one child for every pair $(a,t) \in \TimerAlphabet \times \Qpos$.
	Moreover, $\strategyM$ may behave in a highly non-uniform in its responses to different $(a,t)$ pairs.

    Intuitively, we make use the fact that the winning condition~$\A$ can distinguish only finitely many points.
    This suggests that a more uniform winning strategy should exist: one that plays the same move over large
	intervals of possible timestamps, up to a small number of distinguished singletons.
    We show how to construct such a strategy using the regions of~$\A$, and how to bound the number of clocks and
	the memory required to implement it, thereby obtaining a controller for~\Monitor.

    \para{Step 1: reduce to a finite set of clocks}
	A strategy has, a priori, an infinite number of clocks: one per level of the underlying tree.
	Bounding the number of clocks in such a strategy consists in making the response of \Monitor on an interval of
	timestamps more uniform: to this aim we use the mechanism of regions.
	More precisely, a strategy with $k$-clocks uses only $k$ clocks to monitor and make decisions on
	the time.
	Formally, a \emph{strategy with $k$-clocks} is an extension of $\DTA_k$ whereas the number of locations
	(and thus transitions) may be infinite.

	To restrict the number of clocks, we define an equivalence relation between winning configurations
	according to the timestamps played by \Timer when we have fixed the two letters
	$(a,b) \in \TimerAlphabet \times \MonitorAlphabet$ played by each player.
	Given a timed word $w$, two letters $(a,b)$, and a timestamp $t \in \Qpos$, we say that $(a,b,t)$ is
	a winning continuation of $w$ if there exists a strategy that wins from $w \cdot (a,b,t)$.
	Then for a fixed $(a,b)$ we partition the set of winning timestamps $t$ into a finite union of
	intervals of size greater than $1$, a finite union of singletons, and at most one interval of size
	less than $1$ starting from the last timestamp of $w$ (as illustrated in~\cref{fig:equivalence-part_timestamps}).
	In particular, we obtain the following lemma.

	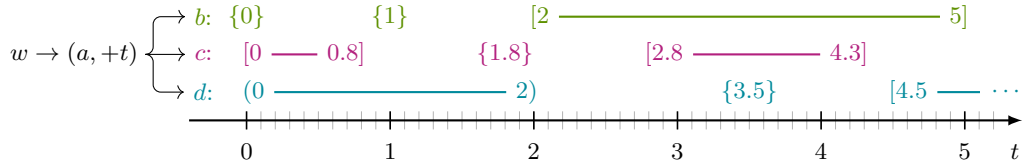
\begin{figure}
		\begin{center}\begin{tikzpicture}[
    semithick,
    every state/.style={minimum size=12pt,inner sep=0},
    every node/.append style={font=\small},
    yscale=1,
    xscale=1.89,
    initial text={},
    every initial by arrow/.append style={inner sep=0, outer sep=0}
]

    \def\tstart{0}
    \def\tend{5}

    \foreach \x in {0,0.1,...,5.3} {
        \draw[black!30, line width=0.3pt] (\x,0.12) -- (\x,-0.12);
    }

    \foreach \x in {0,1,...,\tend} {
        \draw[black, line width=0.6pt] (\x,0.12) -- (\x,-0.12);
        \node[below=2pt, font=\small] at (\x,-0.12) {$\x$};
    }
    \node[below=2pt, font=\small] at (\tend+0.35,-0.12) {$t$};
    \draw[thick,->] (\tstart-0.4,0) -- (\tend+0.4,0);

    \begin{scope}
        \clip (-0.2,0) rectangle (5.1,1.6); 

        \newcommand{\drawInterval}[6]{
            \node[#3,anchor=west,xshift=-0.18cm] (ifrom) at (#1,#4) {$#5#1$};
            \node[#3,anchor=east,xshift=0.18cm] (ito) at (#2,#4) {$#2#6$};
            \draw[thick,#3] (ifrom) -- (ito);
        }
        \drawInterval{0}{2}{c1!70!black}{0.375}{(}{)}
        \drawInterval{4.5}{6}{c1!70!black}{0.375}{[}{]}
        \node[c1!70!black] at (3.5,0.375) {${\{3.5\}}$};

        \drawInterval{0}{0.8}{c2!70!black}{0.875}{[}{]}
        \drawInterval{2.8}{4.3}{c2!70!black}{0.875}{[}{]}
        \node[c2!70!black] at (1.8,0.875) {${\{1.8\}}$};

        \drawInterval{2}{5}{c3!70!black}{1.375}{[}{]}
        \node[c3!70!black] at (1,1.375) {${\{1\}}$};
        \node[c3!70!black] at (0,1.375) {${\{0\}}$};
    \end{scope}
    \node[c3!70!black] at (-.3, 1.375) (b) {$b$:};
    \node[c2!70!black] at (-.3, .875) (c) {$c$:};
    \node[c1!70!black] at (-.3, .375) (d) {$d$:};
    \node[c1!70!black] at (5.3,0.375) {$\cdots$};
    \node[] at (-1.2, .875) (past) {$w \rightarrow (a, {+}t)$};
    \coordinate (center) at (-0.65,0.875);
    \draw[-To,thin,rounded corners=1mm] (past) -- (center) |- (b);
    \draw[-To,thin] (past) -- (c);
    \draw[-To,thin,rounded corners=1mm] (past) -- (center) |- (d);
\end{tikzpicture}%
\end{center}
		\caption{An example showing a structure of the intervals in which a letter from $\MonitorAlphabet = \{b,c,d\}$ is a winning response when \Timer plays $a$ with delay $t \in \Qpos$.
		Note that since \Monitor has a winning strategy,
		if it conforms to it, then at all times it will have a winning response
			for all delays $t \in \Qpos$.}
		\label{fig:equivalence-part_timestamps}
	\end{figure}

	\begin{lemma}
		\label{lemma:equivalence-infinite_kclocks}
		If $\strategyM$ is winning for \Monitor in $\Game$, there exists a winning strategy for
		\Monitor\ with at most $k$ clocks
		where $k = 2|\TimerAlphabet| \cdot K$, and $K$ is the biggest constant in guards of $\A$.
	\end{lemma}

	To prove the previous lemma, we start by introducing the notion of \emph{tight} and \emph{large} continuations
	for two letters according to the size of the interval of winning timestamps.

	\begin{definition}
		Let $w \in (\TimerAlphabet \times \MonitorAlphabet \times \Qpos)^*$,
		and $(a, b, t) \in \TimerAlphabet \times \MonitorAlphabet \times \Qpos$ be a timed letter.
		\begin{enumerate}
			\item $w$ is \emph{winning for \Monitor} if there exists a winning strategy
			for \Monitor to complete $w$.
			\item $(a, b, t)$ is a \emph{winning continuation of $w$ for \Monitor}
			when $w \cdot (a, b, t)$ is winning for \Monitor.
		\end{enumerate}
	\end{definition}

	By letting $\winContMonitor(w)$ be the set of winning continuations of $w$ for \Monitor,
	$w$ be a winning word for \Monitor, and $(a, b, t)$ be a winning continuation of $w$ for \Monitor, we fix
	\begin{align*}
		t_{\inf}(w, a, b, t) &= \inf \{t' \mid \forall t'' \in [t', t], (a, b, t'') \in \winContMonitor(w) \} \\
		t_{\sup}(w, a, b, t) &= \sup \{t' \mid \forall t'' \in [t, t'], (a, b, t'') \in \winContMonitor(w) \}\\
		\intsize(w,a,b,t) &= t_{\sup}(w, a, b, t) - t_{\inf}(w, a, b, t)
	\end{align*}
	Now, we are equipped to define tight and large continuations of a timed word.

	\begin{definition}
		Let $w$ be a winning timed word over $(\TimerAlphabet \times \MonitorAlphabet)$ for \Monitor, and
		$(a, b, t) \in \TimerAlphabet \times \MonitorAlphabet \times \Qpos$ be a winning continuation of $w$
		for \Monitor.
		We say that the continuation $(a, b, t)$ is \emph{tight} if $\intsize(w,a,b,t) = 0$,
		and \emph{large} if $\intsize(w,a,b,t) \geq 1$.
	\end{definition}

	Then, given a timed word $w$, we remark that a winning continuation of $w$
	for \Monitor\ is always tight, large or too close of the last timestamp of $w$.

	\begin{lemma}
		\label{lem:constroller_Pu_oneclock-tight}
		Let $w$ be a winning timed word over $(\TimerAlphabet \times \MonitorAlphabet)$ for \Monitor,
		and $(a, b, t) \in \winContMonitor(w)$.
		Then, $t_{\sup}(w, a, b, t) - t_n  \leq 1$ and $t_{\inf}(w,a ,b, t) = t_n$,
		or $(a, b, t)$ is a tight or a large winning continuation of $w$.
	\end{lemma}
	\begin{nestedproof}
		Let $t$ be a timestamp such that $(a,b,t)$ is a winning continuation that is neither tight nor large,
		i.e., such that $0 < \intsize(w,a,b,t) < 1$.
		We distinguish two cases.

		First, suppose that the continuation is close to $t_n$ but does not start at $t_n$,
		i.e., $t_{\sup}(w, a, b, t) - t_n  \leq 1$ and $t_{\inf}(w,a ,b, t) > t_n$.
		Since, $(a, b, t_n)$ is not part of the winning continuation of $t$,
		there is some $t_n \geq t' < t$ such that $(a,b,t')$ is not a winning continuation of $w$.
		Thus, when $(a,b,t)$ is played, in $\A$ a different guard from the guard used when $(a,b,t_n)$ is played occurs.
		There are two cases for this new guard:
		\begin{itemize}
			\item The guard is an equality constraint, so $\intsize(w,a,b,t) = 0$ and we get a contradiction.

			\item Otherwise, the guard holds for an open interval of valuations.
			Thus, either $\intsize(w,a,b,t)  \geq 1$ or $t_{\inf}(w,a ,b, t) = t_n$,
			both of which give a contradiction.
		\end{itemize}

		Second, suppose that the continuation is far from $t_n$, 
		i.e., $t_{\sup}(w, a, b, t) - t_n  > 1$.
		Since $t_{\sup}(w, a, b, t) - t_n  > 1$, the region of the clock $x$ of the winning condition
		$\A$ has changed along the reading of the new timed letter, and we denote by $r$ this
		new region\footnote{We use the classical notion, introduced in~\cite{Alur1994}, of \emph{regions}
		for $\NTA$ that are classes of equivalence of $\Rpos^{\Cl}$ that are preserved by the reachability relation,
		and so by language acceptance.
		Note that in $\NTA_1$, regions are unit intervals $(n, n+1)$ with $n < K$, singletons $\{n\}$ with $n < K$, or infinite intervals
		$(K, +\infty)$ where $K$ is the greatest constant that appear in the $\NTA_1$.}.
		\begin{itemize}
			\item If $r$ is an open region (i.e., an interval $(n,n+1)$ for some $n < K$, or $(K, +\infty)$
			where $K$ is the biggest constant in $\A$),
			and if $t'$ is a timestamp such that the reading of $(a, b, t')$ reaches $r$, 
			we have, by the equivalence relation given by the regions,
			$t_{\inf}(w, a, b, t) \leq t' \leq t_{\sup}(w, a, b, t)$, in other words $t'$ is is the winning continuation of $t$.
			Also since the size of $r$ is at least one, so is the size of the set of $t'$ that reach $r$.
			From this, we deduce that the winning continuation of $t$ is large.
	
			\item Otherwise, $r$ is a singleton.
			In this case, either one of the open region next to $r$ contains a winning timestamp for $(a, b)$,
			and we apply the previous item, or none of them has a winning timestamp and the winning
			continuation is tight.
		\end{itemize}
		In both cases, we obtain a contradiction.
	\end{nestedproof}

	To conclude the proof of \cref{lemma:equivalence-infinite_kclocks}, we want to show that it is enough
	to keep track of a finite number of tight continuations and that all other can be replaced by a large
	continuation with a different letter of $\MonitorAlphabet$.
	Formally, with $(a,b,t)$ a tight winning continuation of $w$, we say it is \emph{irreplaceable}
	if there is no $b' \in \MonitorAlphabet$ such that $(a,b',t)$ is a large winning continuation.
	
	\begin{lemma}
		\label{lem:equivalence-infinite_nb-tight}
		Let $w$ be a winning timed word over $\TimerAlphabet \times \MonitorAlphabet$ for \Monitor with last timestamp $t_n$.
		Then, there is at most $|\TimerAlphabet| \cdot K$ irreplaceable tight winning continuations,
		where $K$ is the biggest constant in guards of $\A$.
	\end{lemma}
	\begin{nestedproof}
		Intuitively, a tight winning continuation occurs when there is either an equality guard on the transition
		applied, or the current valuation is on the border of the guard on the transition applied.
	
		Now, we note that after $t_n + K$, there exists a large winning continuation on the whole interval
		$(t_n + K,+\infty)$ for some letter $b' \in \MonitorAlphabet$.
		Thus, all irreplaceable tight continuations are between $t_n$ and $t_n + K$.
		Moreover, \Monitor wins in this game; thus for all $(a, t) \in \TimerAlphabet \times \Qpos$ there exists
		some $b \in \MonitorAlphabet$ such that $(a,b,t)$ is a winning continuation.
		To conclude, we prove that if $(a,b,t)$ and $(a,b',t')$ 
		(same $a \in \TimerAlphabet$, different $b,b' \in \MonitorAlphabet$ and $t,t' \in \Qpos$)
		are two tight continuations then $t$ and $t'$ must have distance at least one time unit.
		This ensures that there are at most $K$ tight continuations per letter $a \in \TimerAlphabet$.

		By contradiction, assume there exists two distinct timestamps $t$ and $t'$ with a distance less than $1$
		such that $(a,b,t), (a,b',t')$ are irreplaceable tight continuations.
		But, as stated earlier, for every timestamp $t''$ between $t$ and $t'$, there exists $b'' \in \MonitorAlphabet$
		such that $(a,b'',t'')$ is a winning continuation.
		If one of those continuations is large, then it must include at least one of $t$ and $t'$.
		Say $t$ is included, then either $b'' \neq b$ and then $(a,b,t)$ is not irreplaceable, 
		or $b'' = b$ and then $(a,b,t)$ is not tight, and we obtain a contradiction.
		Thus, suppose that for all $t'' \in [t, t']$, all winning continuations of the form $(a,b'',t'')$ are tight.
		As $\MonitorAlphabet$ is finite, there exist some $b \in \MonitorAlphabet$ such that
		there is an infinite number of timestamps $t_b \in [t, t']$ such that $(a, b, t_b)$ is a tight
		winning continuation.
		Moreover, we note that to separate the two timestamps within one time unit into two distinct regions,
		a guard has to separate the valuations of $x$ after applying each timestamp respectively.
		Thus, as there exists only a finite number of transitions, we deduce the existence of $t_1 < t_2 \in [t, t']$
		such that $(a, b, t_1)$ and $(a, b, t_2)$ are two tight winning continuation and such that the set of regions
		reached by $w \cdot (a, b, t_1)$ and $w \cdot (a, b, t_2)$ contain the same open region.
		Thus, as all timestamps reaching this open region define a winning continuation for $(a, b)$,
		we can conclude that $(a, b, t_1)$ is a large continuation.
		Thus we get a contradiction with the same argument as before.
	\end{nestedproof}
	
	Now we can conclude the proof of \cref{lemma:equivalence-infinite_kclocks}.
	\begin{nestedproof}[Proof of \cref{lemma:equivalence-infinite_kclocks}]
		By Lemma~\ref{lem:equivalence-infinite_nb-tight}, we know that at every step in the game,
		there exists at most $k = 2|\TimerAlphabet| \cdot K$ interesting timestamps, i.e., the irreplaceable tight continuations
		(where the decision of \Monitor may change for a single point),
		as well as upper ($t_{\sup}$) and lower bounds ($t_{\inf}$) of large winning continuations.
		In particular, after each play of \Timer, \Monitor has to save at most $k$ new configurations
		to take its choice, and that can be done with at most $k$ clocks.
		\end{nestedproof}
	
	\para{Step 2: playing on regions}
	From now on suppose that \Monitor has a winning strategy using at most $k$ clocks.
	From this strategy, we define a winning controller
	by limiting first the branching degree of each location (via region equivalence),
	and then the size of each branch (via well quasi orders).

	First, we define the notion of \emph{region-based strategy}.
	To do it, we introduce an (equivalence) relation between two timed words according to s$\A$.
	We denote by $\A(w)$ the set of region-configurations (pairs of location, region) reachable by $\A$ 
	on reading $w$ (this set is not necessarily a singleton since $\A$ is not deterministic).
	Let $w$ and $w'$ be two timed words, we say that $w$ and $w'$ are \emph{equivalent for $\A$}, 
	denoted by $w \sim_\A w'$, when $\A(w) = \A(w')$.
	
	\begin{definition}
		A $k$-clock strategy $\strategy$ is called \emph{region-based} if for all $w$ and $w'$ timed words
		such that $w \sim_\A w'$, and for all $(a, t) \in \TimerAlphabet \times \Qpos$ choices of \Timer,
		we have $\strategy(w, (a, t)) = \strategy(w', (a, t))$.
	\end{definition}
	
	Now, we prove that \Monitor can win with such a strategy.
	
	\begin{lemma}
		\label{lem:equivalence-infinite_region-based}
		If \Monitor has a winning $k$-clocks strategy in $\Game$, then it has a winning region-based $k$-clock strategy.
	\end{lemma}
	\begin{nestedproof}
		Let $\strategy$ be a winning $k$-clocks strategy that is not region-based.
		For each timed word $w$, we let $[w]$ be the representative of the equivalence class of $w$ according to $\sim_\A$.
		We define $\strategy'$ such that $\strategy'(w, (a, t)) = \strategy([w], (a, t))$.
		By definition, $\strategy'$ is a region-based $k$-clocks strategy.
	
		To conclude, we prove that $\strategy'$ is winning.
		By contradiction, suppose there exists a winning timed word $w$ 
		and a choice of \Timer\ $(a, t) \in \TimerAlphabet \times \Qpos$ such that
		$(a, \strategy'(w, (a, t)), t)$ is not a winning continuation.
		However, since $\strategy$ is winning, we have that $(a, \strategy(w, (a, t)), t)$ is a winning continuation.
		In particular, $\strategy(w, (a, t))$ is not necessarily the same as $\strategy'(w, (a, t))$, 
		but $\strategy'(w, (a, t)) = \strategy([w], (a, t))$.
		Moreover, $w \sim_\A [w]$ and regions preserve winning continuations.
		Therefore $(a, \strategy([w], (a,t)), t)$ must be winning as well, which is a contradiction.
	\end{nestedproof}
	
	Now, we can conclude the proof \cref{prop:equivalence-infinite}.
	
	\begin{nestedproof}[Proof of \cref{prop:equivalence-infinite}]
		Assume \Monitor has a winning strategy in $\Game$, and let $K$ be the biggest constant in guards of $\A$.
		In particular, \Monitor admits a winning $k$-clock region-based winning strategy $\strategyM$ with
		$k = 2|\TimerAlphabet| \cdot K$, by \cref{lemma:equivalence-infinite_kclocks,lem:equivalence-infinite_region-based}.
	
		The only remaining step to conclude the proof is to prune infinite branches of $\strategyM$ while
		keeping the winning property, which will yield a winning $k$-clocks controller.
		To do so, we remark that a $k$-clock region-based strategy defines a well-quasi order
		over its configurations (that we can limit to $k$).
		Thus, for every branch, when a decision is given, we reach a smaller set of configurations according
		to this order.
		As this order does not have infinite antichain, all branches end or reach another part of the tree.
		Thus, the strategy obtained is a finite tree that define \DTA\ with outputs.

		Finally, we prove that this controller is winning.
		By contradiction, we suppose that there exists a strategy for \Timer such that the timed word
		produced does not belong to $\Lang(\A)$.
		In particular, against this strategy, the controller play a no longer winning continuation.
		This fact happens when the controller reach a new part of the tree of $\strategyM$.
		Otherwise, we contradict the fact that $\strategyM$ is winning.
		But, this operation is given by the well quasi order, in sense where the set of
		reached configurations by the controller preserved the acceptance condition.
		Thus, we obtain a contradiction with the fact that $\strategyM$ is winning.
	\end{nestedproof}
	This ends the proof of Theorem~\ref{thm:equivalence}.
\end{proof}

\noindent
In all the remaining cases the two synthesis problems do not coincide, namely
existence of a winning strategy does not imply existence of a controller:
first, when 
\Timer is both the \owner of the winning condition, and also the \subject of
the synthesis problem, but the winning condition is an arbitrary $\NTA_1$ language;
second, when the \owner is different from the \subject, even in the case
of reachability $\NTA_1$ winning conditions:

\begin{restatable}{theorem}{thmNonEquivalence}
    \label{thm:nonequivalence}
    The two synthesis problems may not coincide in the following cases:
    \begin{enumerate}
        \item $\Owner\neq\Subject$ and $\A \in$ $\rNTA_1$.
        \item $\Owner=\Subject=\Timer$ and $\A \in\NTA_1$.
    \end{enumerate}
\end{restatable}

%
\begin{proof} 
\fbox{\textbf{1.}}
We start with the case $\Owner=\Timer$ and $\Subject=\Monitor$.
Let \Timer's alphabet be trivial (singleton), and hence 
omitted below; \Timer thus essentially chooses only timestamps.
Let \Monitor's alphabet $\MonitorAlphabet = \{a,b\}$, and let
\Timer's winning condition contain timed words containing some letter ($a$ or $b$)
at distance 1, that is timed words that contain both $(a, t)$ and $(a, t+1)$, or 
both $(b,t)$ and $(b,t+1)$, for some $t\in\Qpos$.
The condition is readily seen to be recognised by a reachability $\NTA_1$.
\Monitor has a winning strategy in this game: it wins by always playing the same letter (say $a$) \emph{unless} timestamp $t-1$ appeared earlier in the play, in which case it plays the opposite of the letter that was played there.
On the other hand, \Monitor has no winning controller, as winning
requires storing unboundedly many different timestamps
(together with letters used at them).

For the case $\Owner=\Monitor$ and $\Subject=\Timer$, 
we choose both alphabets to be trivial (and hence omitted below).
The play is thus just a monotonic sequence 
$t_1 \leq t_2 \leq \cdots$ of timestamps, chosen by \Timer.
Let \Monitor's winning condition consist of those sequences which either exceed 1 at some point
($t_i>1$ for some $i\in\Npos$), or repeat a timestamp 
($t_i = t_{i+1}$ for some $i\in\Npos$).
Both conditions, and hence their union, are recognised by reachability $\NTA_1$.
Furthermore, \Timer has a winning strategy by producing a Zeno word bounded by 1, but
\Timer has no winning controller. Indeed, in order to win, \Timer
should use strictly positive delays only, and therefore, since every controller uses only
finitely many different delays, the bound 1 is inevitably exceeded.

\fbox{\textbf{2.}}
Consider the case $\Owner=\Subject=\Timer$.
We use the same idea as in the previous case: 
\Timer can win only by producing a Zeno word.
Let \Timer's alphabet be trivial (hence omitted below);
\Timer thus essentially chooses only timestamps.
Let \Monitor's alphabet be $\MonitorAlphabet = \{\ok, \err\}$, and let
\Timer's winning condition be the union of two sets:
the timed words 
$w = (b_1, t_1) \, (b_2, t_2)  \cdots \in (\MonitorAlphabet\times\Qpos)^\omega$
that never exceed 1 labelled exclusively by
$\ok$, i.e., $t_i < 1$ and $b_i = \ok$ for all $i\in\Npos$;
and the timed words such that for some $i\in\Npos$ we have
$t_i < t_{i+1}$, $b_1 = \cdots = b_{i} = \ok$ and $b_{i+1} = \err$.
\Timer has a winning strategy by producing a \emph{strictly monotonic} 
Zeno word bounded by 1.
On the other hand a violation of strict monotonicity, namely the equality $t_i=t_{i+1}$, 
is punished by
\Monitor playing $b_{i+1} = \err$.
Therefore, \Timer has no winning controller for the same reason as in the previous case.
\end{proof}

\para{Summary of results}
As our main contribution, we obtain undecidability in all cases of both the reactive and Church synthesis problems,
even under the very restricted case where the timed winning condition is specified by an $\NTA_1$
(or by a $\rresNTA_1$ in one of the variants, see Table~\ref{tab:results} for a summary).
Our reductions in the case of the reactive synthesis problem work uniformly, regardless of which player is the \Subject.
Therefore Table~\ref{tab:results} has six rather than eight entries.
These six cases require four different reductions, and hence
each table entry indicates the problem we reduce from.
When \Timer is the \owner of the winning condition, some undecidable problems for
lossy counter machines (LCM) turn out suitable for reductions,
namely repeated reachability and boundedness~\cite{Schnoebelen10a} 
(these results are the core technical part);
and when \Monitor is the \owner, we use two variants of universality of $\NTA_1$ languages
(these reductions are comparatively simpler).

\begin{table}[h]
    \centering
    \begin{tabular}{|C{1.15cm}|C{3.45cm}|C{4.05cm}|C{3.55cm}|}
        \cline{2-4}
        \multicolumn{1}{c|}{} &
        \multicolumn{1}{c|}{Reactive synthesis} & \multicolumn{2}{c|}{Church synthesis}  \\
        \cline{2-4}
        \multicolumn{1}{c|}{} &
        \multicolumn{1}{c|}{any \Subject}
        & \multicolumn{1}{c|}{$\Subject = \Timer$} & \multicolumn{1}{c|}{$\Subject = \Monitor$}  \\
        \hline\Owner\par\raisebox{0.2ex}{\rotatebox[origin=c]{90}{$=$}}\par\Timer
        & $\A\in\NTA_1$ \par LCM repeated reach. \par (\cref{thm:Buchi}, Section~\ref{sec:Buchi})
        & $\A\in\NTA_1$ \par LCM repeated reach. \par (\cref{thm:Buchi}, Section~\ref{sec:Buchi})
        & $\A\in\rresNTA_1$ \par LCM boundedness \par (\cref{thm:Radek}, Section~\ref{sec:boundedness}) \\
        \hline\Owner\par\raisebox{0.2ex}{\rotatebox[origin=c]{90}{$=$}}\par\Monitor
        & $\A\in\NTA_1$\par $\NTA_1$ universality \par (\cref{thm:universality}, Section~\ref{sec:universality})
        & $\A\in\NTA_1$\par $\NTA_1$ sampled universality \par (\cref{thm:universality-sample}, Section~\ref{sec:universality})
        & $\A\in\NTA_1$\par $\NTA_1$ universality \par (\cref{thm:universality}, Section~\ref{sec:universality}) \\
        \hline
    \end{tabular}
    \caption{Summary of our undecidability results.}
    \label{tab:results}
\end{table}


\noindent
The table specifies the class of winning conditions $\A$ for which we prove undecidability:
$\A\in\NTA_1$, except for one exception
where we need \emph{1-resetting} reachability $\NTA_1$, a slight extension of reachability $\NTA_1$
that uses a limited form of $\varepsilon$-transitions that reset the clock whenever it reaches value 1.
%
We believe that resorting to $\resNTA_1$ does not weaken our results, for the following two reasons:
first, while $\NTA_1$ with $\varepsilon$-transitions correspond to 
nondeterministic one-register automata ($\NRA_1$) \emph{with guessing}
(see \cite{FHL16}),
$\resNTA_1$~still correspond to $\NRA_1$ \emph{without guessing};
second, all our undecidability results translate from the timed setting to the data setting, where the 
winning sets are specified by $\NRA_1$ without guessing.

\section{\boldmath\Monitor is the \owner}
\label{sec:universality}

As a warm-up, we prove undecidability of both synthesis problems in the case where
\Monitor owns the winning condition, by reduction from two variants of the universality problem for $\NTA_1$
(shown undecidable in~\cite{Lasota2008} and~\cite{Abdulla2007}).

\begin{restatable}{theorem}{thmUniversality}
    \label{thm:universality}
    When $\Owner = \Monitor$, the following problems are undecidable:
    \begin{enumerate}
        \item the timed reactive synthesis, irrespectively of \Subject; 
        \item the timed Church synthesis, when $\Subject=\Monitor$.
    \end{enumerate}
\end{restatable}

\begin{proof} 
   We reduce from the $\NTA_1$ universality problem 
   (does the language of a given $\A\in\NTA_1$ contain all infinite timed words?)~\cite{Lasota2008}.
   Given an $\NTA_1$ $\A$ over alphabet $\Sigma$, we construct a
    timed game where     
    \Monitor's alphabet is the singleton $\{\square\}$, \Timer's alphabet is~$\Sigma$, and
    the winning condition of \Monitor is $\Lang(\A)$ (ignoring the letters ``$\square$'').
    Therefore, \Timer\ wins if it produces a timed word over $\Sigma$ that is not in $\Lang(\A)$.
    Thus, \Timer has a winning strategy if $\Lang(\A)$ does not contain all timed words, and \Monitor
     has 
     a winning controller
    otherwise (the trivial memoryless one that keeps emitting ``$\square$'').
%
\end{proof}

\begin{remark}
    The reduction does not guarantee the existence of a controller for \Timer, because \Timer's strategy may need to produce an infinite timed word.
    However, in the case of $\rNTA_1$ the strategy only needs to produce a finite prefix, and therefore
    we obtain HyperAckermann-hardness of
    both synthesis problems, regardless of \Owner and \Subject.
\end{remark}

For  undecidability of the timed Church synthesis problem when $\Subject=\Timer$, we turn to
the \emph{sampled} universality problem, a variant that assumes
the \emph{sampled} semantics $\Lang_\delta(\_)$  of timed automata.
Given $\delta >0$, we say that a timed word has \emph{granularity} $\delta$ if
all its timestamps are multiplicities of $\delta$,
and define $\Lang_\delta(\A)\subseteq\Lang(\A)$ as the language of all timed words
of granularity $\delta$ in $\Lang(\A)$. 
Given an \NTA $\A$,
the \emph{sampled} universality problem asks if
for all $\delta > 0$, the language $\Lang_\delta(\A)$ contains all timed words
of granularity $\delta$.
The problem
(called \emph{universal sampled universality} in~\cite{Abdulla2007})
is known to be undecidable  for $\NTA_1$~\cite[Thm.~3]{Abdulla2007}.

\begin{restatable}{theorem}{thmUniversalitySample}
    \label{thm:universality-sample}
    The timed Church synthesis problem is undecidable when $\Owner = \Monitor$ and $\Subject = \Timer$.
\end{restatable}

\begin{proof} 
    We proceed similarly as in \cref{thm:universality}, but reduce from the $\NTA_1$ sampled universality problem.
       Given an $\NTA_1\ \A$ over alphabet $\Sigma$, we construct the same
    timed game as in \cref{thm:universality}:
    $\MonitorAlphabet=\{\square\}$,  $\TimerAlphabet=\Sigma$, and
    the winning condition of \Monitor is $\Lang(\A)$ (ignoring ``$\square$'').
    
    If \Timer has a winning controller $\C$ in this game, we take an arbitrary rational $\delta>0$ such that
    all positive delays used by $\C$ are multiplicities of $\delta$, and deduce that
    any infinite timed word output by $\C$ has granularity $\delta$ and is not in $\Lang_\delta(\A)$.
    
    In the opposite direction, suppose a timed word $w$ of granularity $\delta$ is not accepted
    by $\A$.
    Assume, \wlogg, that $\A$ has a run labeled by $w$.
    We build a word $w'$ whose run visits only finitely many
    different configurations.
    Let $M$ be the maximal constant appearing in the guards of $\A$.
    If the valuations in the run labeled by $w$ never exceed $M$, i.e. if $\A$ always resets the clock when it goes over $M$,
    then we take $w' = w$ and we are done because every valuation is always a multiplicity of $\delta$ and there are finitely many 
    thereof between $0$ and $M$.
    Otherwise, let $t_M > M$ be the first valuation in the run that exceeds $M$.
    We obtain $w'$ by modifying $w$ so that every time the valuation in the run is over $M$, it is exactly $t_M$
     (using repeating timestamps if necessary),
    but the run is unchanged when the valuation is at most $M$.
%
    Since valuations over $M$ are indistinguishable by $\A$,
    $w'$ is not accepted by $\A$.
%
    Therefore no accepting locations appear on the run
    from some point on, and hence
    the run necessarily forms some loop without accepting locations on it.
    We replace $w'$ by the timed word $w''$ obtained by letting $\A$ repeat the loop infinitely. 
    Like $w'$, the word $w''$ has granularity $\delta$ and is not in $\Lang_\delta(\A)$,
    but $w''$ is additionally ultimately periodic when seen as a sequence of letter-delay pairs 
    $w'' = (a_1, \tau_1) \, (a_2, \tau_2) \cdots$.
    In consequence, \Timer has a winning controller that produces $w''$.
\end{proof}

\section{Lossy counter machines}
\label{sec:lossy}

We now introduce the model to be used in the forthcoming reductions.
A \emph{lossy counter machine} (LCM) is a tuple 
$\M = \tuple{\Counter, \States, \state_0, \Instructions}$,
where $\Counter$ is a finite set of counters, $\States$ is a finite set of control locations,
$\state_0 \in \States$ is the initial control location, and $\Instructions$ is a finite set of instructions
$\instruction = (\state, \op, \state')$, where $s,s' \in \States$ and
$\op \in \{c\increment, c\decrement,c\ztest \mid \counter \in \Counter\} $ 
is the operation of increment, decrement, or zero test on some counter $c\in\Counter$.
%
We let $\Instructions_\inc^c, \Instructions_\dec^c, \Instructions_\zt^c$ denote
the sets of instructions incrementing, decrementing, and zero testing counter $c$, respectively.
We also let $\Instructions^c = \Instructions_\inc^c \cup \Instructions_\dec^c \cup \Instructions_\zt^c$
be the set of all instructions affecting $c$.

A \emph{configuration} of an LCM $\M$ is a pair $(\state, \val) \in \States \times \N^{\Counter}$
where $\state$ is a control location, and $\val$ is a counter valuation.
Given counter valuations $\mu, \val \in \N^{\Counter}$,
we write $\mu \leq \val$ whenever $\mu(\counter) \leq \val(\counter)$ for every $\counter\in \Counter$.
We define two equivalent
 semantics: the \emph{lossy} semantics, and the \emph{free-test} one
 (corresponding, intuitively, to postponing the losses until zero tests).
A \emph{run} of an LCM under either semantics is a maximal (finite or infinite) sequence of configurations
starting from the initial configuration $(\state_0, \val_0)$, where $\val_0 = \zeroval$, such that
for every two consecutive configurations $(\state, \val)$ and $(\state', \val')$
there exists an instruction $\instruction = (\state, \op, \state') \in \Instructions$ such that
$(\state, \val) \xrightarrow{\instruction} (\state', \val')$.
Under the lossy semantics, $(\state, \val) \xrightarrow{\instruction} (\state', \val')$ holds if
%
\begin{enumerate}
\item $\instruction \in \Instructions_\inc^c$ and $\val' \leq \val[c\increment]$, where $\val[c\increment]$
is as $\val$ except that $c$ is incremented by 1; or
\item $\instruction \in \Instructions_\dec^c$ and $\val' \leq \val[c\decrement]$, where $\val[c\decrement]$
is as $\val$ except that $c$ is  decremented by 1; or
\item $\instruction \in \Instructions_\zt^c$, $\val(c) = 0$, and $\val' \leq \val$.
\end{enumerate}

\noindent
Under the free-test semantics,
 $(\state, \val) \xrightarrow{\instruction} (\state', \val')$ holds if
\begin{enumerate}
\item $\instruction \in \Instructions_\inc^c$ and $\val' = \val[c\increment]$; or
\item $\instruction \in \Instructions_\dec^c$ and $\val' = \val[c\decrement]$; or
\item $\instruction \in \Instructions_\zt^c$ and $\val' = \val[c \coloneqq 0]$ is as
$\val$ except that $c$ is set to 0.
\end{enumerate}

%
\noindent
One can see the free-test semantics as delaying losses until the next zero test, at which point the counter loses its value entirely before proceeding to the next configuration.

We assume, \wlogg, that $\M$ is \emph{deterministic}, i.e. that for any configuration $(\state,\val)$ there exists at most one instruction $\instruction$ leading to another configuration.
Under the free-test semantics, this implies that a given configuration only has at most one successor.
Under the lossy semantics, there may be a number of successor configurations differing only on the counter valuations (how much is lost for each counter after the instruction).
Let $\MRuns_\M$ denote the runs of $\M$.
We assume, \wlogg, that they are all infinite.

The \emph{boundedness} problem asks whether all runs of a given LCM $\M$ visit only finitely many configurations,
or equivalently whether all runs of $\M$ are bounded in their valuations.
%
The \emph{repeated reachability problem} asks if $\M$ has an infinite run
visiting infinitely often a given control location $\state \in \States$.
Both problems are known to be undecidable
already for 4 counters~\cite[Thms.~10,12]{Mayr03}
(an excellent survey is \cite{Schnoebelen10a}),
and the choice of semantics is irrelevant.

\section{\boldmath\Timer is both the \owner and the \subject}
\label{sec:Buchi}

In this section we prove undecidability of both synthesis problems when \Timer 
owns the winning condition, except when $\Subject=\Monitor$
(investigated in the next section).

\begin{restatable}{theorem}{thmBuchi}
    \label{thm:Buchi}
%
%
    When $\Owner = \Timer$, the following problems are undecidable:
    \begin{itemize}
        \item the timed reactive synthesis, irrespectively of \Subject; 
        \item the timed Church synthesis, when $\Subject=\Timer$.
    \end{itemize}
\end{restatable}

\medskip

\noindent
(Note the symmetry of \cref{thm:Buchi,thm:universality} along the exchange of roles of \Timer and \Monitor.)

In the rest of this section we prove \cref{thm:Buchi} by providing a reduction
from the repeated reachability problem for LCM.
%
%
To this aim we fix a lossy counter machine 
$\M = \tuple{\Counter, \States, \state_0, \Instructions}$
with four counters $\Counter = \{\counter_1, \counter_2, \counter_3, \counter_4\}$
and a location $\state\in\States$,
and construct a timed game with the property that \Timer has a winning controller
if $\M$ repeatedly reaches $\state$, and \Monitor has a winning strategy otherwise
(see \cref{lem:Buchi_correctness}).
In the proof we assume lossy semantics of $\M$.
Our approach is inspired by the proof of~\cite[Thm.~8.4]{Piorkowski-PhD_22}
(we note however substantial differences: 
timed Church synthesis was considered there with
\Timer's  winning condition specified by 
$\NTA_2$ while we restrict ourselves to $\NTA_1$, and moreover, there
\Monitor's controller was sought, while we seek \Timer's one).

\para{The idea of reduction}
In the course of the game, \Timer is tasked with producing an increasingly longer timed word supposed
to be an encoding of a run of $\M$.
However, \Timer may also \emph{cheat} and
produce a timed word which contains an \emph{error} and therefore 
is not a correct run encoding.
We distinguish four types of errors.
In order to prevent \Timer from cheating, \Monitor verifies if the encoding proposed by \Timer\ is correct,
and if it is not, \Monitor has to  declare the detected type of error correctly, and
immediately, that is in the same round the error occurs.
Therefore one way of winning by \Timer is to see $\state$ infinitely often while seeing
no error declared by \Monitor (using cheating or not);
or to mislead \Monitor about the correctness of encoding by cheating and either seeing no
immediate error declaration, or seeing an error declaration of wrong type.
On the other hand, when \Monitor manages to declare immediately the correct type of error,
the play is winning for it irrespectively of the continuation.

\para{\boldmath Encoding of runs of $\M$}
A central ingredient of our reduction is the encoding of runs of $\M$ as timed words.
Let $\Sigma_{\mathcal{M}} = \Counter \cup \Instructions$.
We first introduce the \emph{valuation encoding}. For a valuation $\val = (v_1, v_2, v_3, v_4) \in \N^\Counter$ define
$
\enc(\val) = \counter_1^{v_1}\,\counter_2^{v_2}
\counter_3^{v_3}\counter_4^{v_4}
$, a finite untimed word consisting of four \emph{segments}.
The set of all such encodings is $\ValEnc_{\M} = \counter_1^*\counter_2^*\counter_3^*\counter_4^*$.

Next, we define the \emph{run encoding}.
The function $\enc \colon \MRuns_\M \to \Sigma_{\mathcal{M}}^\omega$ maps a run
$\run = (\state_0, \val_0) \xrightarrow{\instruction_1} (\state_1, \val_1) \xrightarrow{\instruction_2}
\cdots$ to the infinite untimed word
$
\enc(\run) = \enc(\val_0) \, \instruction_1 \, \enc(\val_1) \, \instruction_2 \, \cdots
$.
Let $\RunEnc_{\M} = \{ \enc(\run) \mid \run \in \MRuns_\M\}$ be the set of valid untimed encodings.

\begin{figure}
    \begin{center}\begin{tikzpicture}[
    xscale=2.2,
    every node/.append style={font=\small},
]
    \def\tstart{0}
    \def\tend{6}

    \foreach \x in {\tstart,0.1,...,\tend} {
        \draw[black!30, line width=0.3pt] (\x,0.12) -- (\x,-0.12);
    }

    \foreach \x in {\tstart,1,...,\tend} {
        \draw[black, line width=0.6pt] (\x,0.12) -- (\x,-0.12);
        \node[below=2pt, font=\small] at (\x,-0.12) {\x};
    }

    \foreach \x/\l in {0,...,\numexpr\tend-1\relax} {
        \node[font=\small] (q\x) at (\x+0.5,1.6) {$(s_\x,\nu_\x)$};
    }
    \node[font=\small] (q6) at (6.1,1.6) {$\cdots$};

    \foreach \x/\l in {0/{c_1\mkern-2mu\increment},1/{c_2\mkern-2mu\increment},2/{c_1\mkern-2mu\increment},3/{c_1\mkern-2mu\increment},4/{c_1\mkern-2mu\ztest}} {
        \pgfmathtruncatemacro{\y}{\x+1}
        \draw[-To] (q\x) -- node[above] {$\l$} (q\y);
        \coordinate (t\y) at (\y,0.2) {};
        \coordinate (z\y) at (\y,-0.2) {};
        \node[font=\scriptsize, above=6.5mm of t\y,inner sep=0] (l\y) {$(s_\x,\l,s_\y)$};
        \draw[c1] (\y,0.25) -- (l\y);
        \fill[c1] (t\y) -- ++(0.03,0.2) -- ++(-0.06,0) -- cycle;
        \fill[black] (t\y) -- ++(0.0225,0.15) -- ++(-0.045,0) -- cycle;
    }
    \draw[] (q5) -- (q6);

    \draw[thick,->] (\tstart-0.2,0) -- (\tend+0.2,0);

    \def\events{%
            {1/1.3/$c_1$/c2},%
            {2/2.3/$c_1$/c2},%
            {3/2.5/$c_2$/c3},%
            {4/3.3/$c_1$/c2},%
            {5/3.2/$c_1$/c4},%
            {5/3.5/$c_2$/c3},%
            {8/4.5/$c_2$/c3},%
            {9/5.5/$c_2$/c3}%
    }
    \coordinate (lfin) at (7.2,0.2) {};
    \coordinate (bfin) at (7.2,-0.2) {};
    \coordinate (bini) at (1.7,-0.2) {};
    \foreach \id/\x/\label/\c in \events {
        \coordinate (e\id) at (\x,0.2) {};
        \coordinate (b\id) at (\x,-0.2) {};
        \fill[fill=\c] (e\id) -- ++(0.03,0.2) -- ++(-0.06,0) -- cycle;
        \fill[black] (e\id) -- ++(0.0225,0.15) -- ++(-0.045,0) -- cycle;
        \node[font=\scriptsize, above=3mm of e\id,inner sep=0] (l\id) {\label};
    }

    \foreach \x/\c/\d in {0/0/0,1/1/0,2/1/1,3/2/1,4/0/1,5/0/1} {
        \pgfmathtruncatemacro{\y}{\x+1}
        \node at (\x+0.5, -0.9) {\scriptsize$\nu_\x\mkern-3mu=\mkern-3mu(\c,\mkern-2mu\d,\mkern-2mu0,\mkern-2mu0)$};
    }
    \node[anchor=west] at (-0.2,1.6) {run:};
    \node[anchor=west] at (-0.2,0.67) {encoding:};
    \node[anchor=west] at (-0.2,-0.9) {$\nu_i$:};
\end{tikzpicture}\end{center}
    \caption{Illustration of the encoding of runs used in this section.}
    \label{figure:run-block-encoding}
\end{figure}
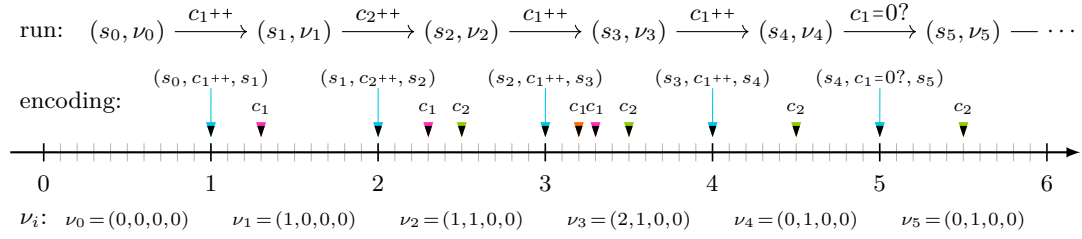

Intuitively, we exploit the timed structure of a word to enforce the correctness of run encodings (see \cref{lem:Buchi_language}).
To this end, we define a timed language $\RunEnc_{\M}^\Time$.
We specify it using an untimed $\omega$-regular language $\RegIM$
which enforces the structural shape of runs of $\M$, and
timed languages $A^\Time_\M$, $B^\Time_\M$ and $C^\Time_\M$,
which impose additional timed properties.
Let $\RegIM$ consist of all (untimed) words
$w \in \Sigma_{\mathcal{M}}^\omega$
that satisfy the following regular conditions:

    \begin{description}
        \item[Block structure:]
                 $w$ is an infinite interleaving of valuation encodings and instructions, i.e., $w$ is of the form
		$w = \enc(\val_0) \, \instruction_1 \, \enc(\val_{1}) \, \instruction_2\, \cdots
		\in (\ValEnc_{\M} \, \Instructions)^{\omega}$ with $\instruction_1$ starting from $\state_0$.

        \item[Instruction compatibility:] Each instruction's target state is the source state of the next one, i.e., for every infix $\instruction_{i} \, \Counter^* \, \instruction_{i+1}$ we have $\instruction_i = (\_, \_, \state)$ and $\instruction_{i+1} = (\state_, \_, \_)$ for some $s$.

        \item[Consistency with zero tests:]
	every infix 
         $\instruction_i \, \enc(\val_{i})\, \instruction_{i+1}$ of $w$ with $\instruction_{i+1} \in \Instructions_\zt^\counter$
	verifies $\enc(\val_{i}) \in (\Counter \setminus \{\counter\})^*$,
        i.e., the symbol $\counter$ does not appear in $\enc(\val_{i})$.
    \end{description}
%

\noindent
Clearly  $\RegIM \supseteq \RunEnc_{\mathcal{M}}$.
%
Let $\RegIMT = \untimed^{-1}(\RegIM)$
contain all timed words whose untiming is in $\RegIM$.
We define $\RunEnc_{\mathcal{M}}^\Time$
as the intersection of four timed languages:
\[
    \RunEnc_{\M}^\Time = \RegIMT \cap \AM
    \cap \BM \cap C_{\mathcal{M}}^{\Time} \subseteq (\Sigma_\M \times \Qpos)^{\omega}
\]
\noindent
where 
$A_{\mathcal{M}}^{\Time}$, $B_{\mathcal{M}}^{\Time}$ and $C_{\mathcal{M}}^{\Time}$
are languages of infinite timed words $w$ satisfying the conditions given below.
For simplicity, we define the condition for $C_{\mathcal{M}}^{\Time}$ (the most intricate), assuming that $w \in L=\RegIMT \cap \AM \cap \BM$, as it is irrelevant how it treats words outside $L$.

    \begin{description}
        \item[Strict monotonicity ($\AM$):]
        $w$ is strictly monotonic (no timestamp repeats).

        \item[Blocks align to unit intervals ($\BM$):]
        \label{itm:buchi_timed-B}
        symbols from $\Instructions$ appear in $w$ with consecutive integer time\-stamps starting from $1$.

        \item[Well-alignment of valuations encodings ($\CM$):]
        every maximal infix of $w$ of the form $\enc(\val_i)$ \, $\instruction_i \, \enc(\val_{i+1})$ verifies the following conditions,
        for all counters $c \in \Counter$:
        \begin{enumerate}
            \item\label{itm:buchi_timed-C-not}
            if $\instruction_i \notin \Instructions_\dec^\counter \cup \Instructions_\inc^\counter$
			then $\val_{i+1}(\counter) \leq \val_{i}(\counter)$, and
            a timed letter $(\counter, t)$ appears in $\enc(\val_{i+1})$ only if 
            $(\counter, t-1)$ appears in $\enc(\val_{i})$;

            \item\label{itm:buchi_timed-C-increase}
            if $\instruction_i \in \Instructions_\inc^\counter$, then $\val_{i+1}(\counter) \leq \val_{i}(\counter) + 1$, and
            the encoding verifies the case~\ref{itm:buchi_timed-C-not}
            except for a (potential) one extra letter $\counter$ appearing in the beginning of
            the $c$-segment in $\enc(\val_{i+1})$ within less than one unit of time
            from the first symbol $\counter$ of $\enc(\val_i)$.

            \item\label{itm:itm::buchi_timed-D-decrease}
            if $\instruction_i \in \Instructions_\dec^\counter$, then
            $\val_{i+1}(\counter) \leq \val_{i}(\counter) - 1$, and
            the encoding verifies the case~\ref{itm:buchi_timed-C-not}
            except for the first letter $\counter$ appearing in $\enc(\val_i)$ which can not appear in 
            $\enc(\val_{i+1})$.
        \end{enumerate}
    \end{description}

\noindent    
According to the definition of $C_{\mathcal{M}}^{\Time}$, every increment 
(resp.~decrement) of a counter $\counter$ results in adding (resp.~removing)
one letter $c$ in the \emph{beginning} of the $c$-segment.
An illustration is provided in \cref{figure:run-block-encoding}.

%
%
\begin{apxlemmarep}
    \label{lem:Buchi_language}
        \label{itm:Buchi_language-2}
        $\RunEnc_{\mathcal{M}} = \{\untimed(w) \mid w \in \RunEnc^{\Time}_{\mathcal{M}}\}$.
\end{apxlemmarep}

\begin{proof}
	\label{app:Buchi_language}
	We reason by double inclusion, and we first prove that $\RunEnc_{\mathcal{M}} \supseteq
	\{\untimed(w) \mid w \in \RunEnc^{\Time}_{\mathcal{M}}\}$.
	Let $w \in \RunEnc^{\Time}_{\mathcal{M}} \subseteq
	\RegIMT$ (by definition of $\RunEnc^{\Time}_{\mathcal{M}}$).
	In particular, we deduce that $\untimed(w) \in \RegIM
	\subseteq \RunEnc_{\mathcal{M}}$ and thus, $\RegIM = \RunEnc_{\mathcal{M}}$.

	Conversely, we prove that $\RunEnc_{\mathcal{M}} \subseteq
	\{\untimed(w) \mid w \in \RunEnc^{\Time}_{\mathcal{M}}\}$.
	Let $w \in \RunEnc_{\mathcal{M}}$, we want to define a timed word
	$w^\Time \in \RunEnc^{\Time}_{\mathcal{M}}$ such that $\untimed(w^\Time) = w$.
	Since $w  \in \RunEnc_{\mathcal{M}} \subseteq \RegIM$, such a $w^\Time$ would by definition be in $\RegIMT$.
	Thus it suffices to verify that $\untimed(w^\Time) = w$ and
	that $w^\Time \in A^\Time_\M \cap B^\Time_\M \cap C^\Time_\M$.
	Building $w^\Time$ requires defining the timestamps
	of letters in $w$ and, more specifically, timestamps for the letters from $\Counter$
	(timestamps for $\Instructions$ are forced by $B_{\mathcal{M}}^{\Time}$).

	Let $\run = (\state_0, \val_0) \xrightarrow{\instruction_1} \, (\state_1, \val_1) \xrightarrow{\instruction_2} \cdots$
	be a run of $\mathcal{M}$ such that $w = \enc(\val_0) \instruction_1 \enc(\val_1) \instruction_2 \cdots$.
	We define a sequence of timed words $\enc^\Time(\val_0), \enc^\Time(\val_1), \dots$ such that
	$\untimed(\enc^\Time(\val_n)) = \enc(\val_n)$ for all $n \geq 0$ recursively as follows.
	For the base case, $\enc^\Time(\val_0) = \varepsilon$.
	Assuming $\enc^\Time(\val_n)$ has been defined, then $\enc^\Time(\val_{n+1})$ is the minimal word such that,
	for all $\counter \in \Counter$:
	\begin{itemize}
	    \item if $\instruction_{n} \notin \Instructions_\inc^\counter \cup \Instructions_\dec^\counter$,
	    then for the first $\val_{n+1}(c)$ timed letters $(\counter, t)$ in $\enc^\Time(\val_n)$,
	    we have $(\counter, t +1)$ in $\enc^\Time(\val_{n+1})$;
	
	    \item if $\instruction_n \in \Instructions_\dec^\counter$,
	    then for every timed letter $(\counter,t)$ in $\enc^\Time(\val_n)$ starting from the second occurrence
		to the $(\val_{n+1}(c)+1)$th occurrence, we have $(\counter,t+1)$ in $\enc^\Time(\val_{n+1})$;
	
	    \item if $\instruction_{n} \in \Instructions_\inc^\counter$, and $\val_{n+1}(c) \leq \val_n(c)$,
	    we apply the first case.
	    Otherwise,  $\val_{n+1}(c) = \val_n(c) + 1$.
	    Then first for every timed letter $(\counter,t)$ in $\enc^\Time(\val_n)$ we have $(\counter,t+1)$ in
		$\enc^\Time(\val_{n+1})$.
	    Moreover, assume $\counter = \counter_i$ for $i \in \{1,2,3,4\}$.
	    We let $t_{\min} = n + (i-1)/4$, and $t_1$ be the timestamp of the first occurrence in $\enc^\Time(\val_n)$ of
		$\counter$, or $n+i/4$ if there is no such occurrence.
	    Then finally we also have $(\counter, (t_{\min} + t_1)/2)$ in $\enc^\Time(\val_{n+1})$.
	\end{itemize}
	One can easily verify that these timed words indeed satisfy $\untimed(\enc^\Time(\val_n)) = \enc(\val_n)$ and that a
	ll timestamps in $\enc^\Time(\val_n)$ are in the interval $(n,n+1)$ for all $n \geq 0$.
	Finally, we let $w^\Time = \enc^\Time(\val_0) \cdot (\instruction_1,1) \cdot \enc^\Time(\val_1) \cdot
	(\instruction_2,2) \cdot \dots$.
	Showing that $w^\Time$ is in $A^\Time_\M$ and $B^\Time_\M$ is trivial, and a case analysis on the instructions
	shows that $w^\Time$ is also in $C^\Time_\M$, thus concluding the proof.
\end{proof}

\para{The timed game}
We define the timed synthesis game where $\Timer$'s actions are 
$\TimerAlphabet  = \Sigma_{\mathcal{M}}$
and $\Monitor$'s actions are $\MonitorAlphabet = \{\Ok, \errR, \errA, \errB, \errC\}$.
For the definition of the winning condition it is crucial that all the languages
$\RegIMT, \AM, \BM$ and $C_{\mathcal{M}}^{\Time}$ can be defined
\emph{locally}: an infinite word $w$ belongs to $\RegIMT$ exactly when all (finite)
prefixes of $w$
belong to a certain local language $\RegIMTloc \subseteq (\Sigma_\M \times \Qpos)^*$
of \emph{finite} timed words, which is moreover recognised by a $\rNTA_1$;
and likewise for the $\AM$, $\BM$ and $C_{\mathcal{M}}^{\Time}$.
Specifically:

    \begin{itemize}
\item $\RegIMTloc$:
the finite prefix satisfies the defining conditions
of $\RegIM$.
\item $\AMloc$:
the last two timestamps are nonequal
(or the word is of length 1, the border case).
\item $\BMloc$: if the last letter is from $\Instructions$,
then no $\Instructions$ appears in the last open unit interval, and $\Instructions$ appears
exactly one time unit before the last timestamp; and
if the last letter is not from $\Instructions$, some 
$\Instructions$ occurs less that one unit before the last timestamp.
(We omit the border case.)
\item $\CMloc$:
the last timed letter of the word, if it is $(\counter,t)\in\Counter\times\Qpos$, appears also one unit of time before
as $(\counter, t-1)$, 
except when it is the first letter in the $\counter$-segment and the last instruction
is an increment of $\counter$.
Moreover, when $(c, t)$ is the first letter in the $\counter$-segment and the last instruction
is a decrement of $\counter$,
the language $\CMloc$
requires that the first letter in the previous $\counter$-segment was removed,
i.e., its timestamp is strictly smaller that $t-1$.
As before, we conveniently
assume that the word is in $L=\RegIMTloc \cap \AMloc
    \cap \BMloc$, as
    it is irrelevant which words from 
    $\widehat{L} = \widehat{\RegIMTloc} \cup \widehat{\AMloc}
    \cup \widehat{\BMloc}$
    belong to $C_{\mathcal{M}}^{\Time,\ell}$.
\end{itemize}
\noindent
Note that $\CMloc$ is the most difficult one, and the only one for which we will need non-determinism.

The winning condition is also defined mostly locally,
by a combination of restrictions imposed on \Timer's or \Monitor's moves.
To this aim we use the projections of finite words,
$\projB \colon (\TimerAlphabet \times \MonitorAlphabet \times \Qpos)^* \to \MonitorAlphabet^*$ and
$\projA \colon (\TimerAlphabet \times \MonitorAlphabet \times \Qpos)^* \to (\TimerAlphabet \times \Qpos)^*$, as well as their inverses
$\projB^{-1}$ and $\projA^{-1}$.
We define \Timer's winning set as
\begin{displaymath}
	W_\mathcal{M} = \Reach(V^{\Time}_\mathcal{M}) \cup
	\Big(\projA^{-1}(\mathrm{Inf}(s)) \cap \projB^{-1}(\{\Ok\}^\omega) \Big)
\end{displaymath}
where 
$\mathrm{Inf}(s)$ is the (untimed) regular language ``the location $s$ appears infinitely often'',
$\Reach(V^{\Time}_\M)$ stands for
the language
of those infinite timed words which have some prefix in $V^{\Time}_\M$, and
$V^{\Time}_\M$ itself is the following union of languages of finite timed words:
\begin{align*}
    V^{\Time}_\mathcal{M} = \ & 
    \big(\projB^{-1}(\Ok^*\,\errR) \cap
    \projA^{-1}(\RegIMTloc)\big) \ \cup {}
    \tag{\small \Monitor wrongly claims an $\RegI$ error}\\
    &
    \big(\projB^{-1}(\Ok^*\,\errA) \cap
    \projA^{-1}(A^{\Time, \ell}_\mathcal{M})\big) \ \cup {}
    \tag{\small \Monitor wrongly claims an $A$ error}\\
    & 
    \big(\projB^{-1}(\Ok^*\,\errB) \cap
    \projA^{-1}(B^{\Time, \ell}_\mathcal{M})\big) \ \cup {}
    \tag{\small \Monitor wrongly claims a $B$ error}\\
    & 
    \big(\projB^{-1}(\Ok^*\,\errC) \cap
    \projA^{-1}(C^{\Time, \ell}_\M \cup \widehat{\RegIMTloc} \cup
                        \widehat{\AMloc}    \cup \widehat{\BMloc})\big) \\
    & \tag{\small \Monitor wrongly claims a $C$ error}
\end{align*}
Thus, \Timer wins if $\state$ is occurring infinitely often and \Monitor declares no error,
i.e., plays exclusively
$\Ok$ moves, or some finite prefix of play is in $V^{\Time}_\M$, namely
\Monitor declares a wrong type of error.
Specifically, \Monitor declares
$\errR$, $\errA$ or $\errB$ in the round when the corresponding local condition holds,
or $\errC$ in the round when either the local condition $C^{\Time, \ell}_\M$ holds,
or some of the other three local conditions fails.
Intuitively, $\errR$, $\errA$ or $\errB$ are prioritised over $\errC$:
in order to win by declaring an error, \Monitor must declare
$\errR$, $\errA$ or $\errB$ if the corresponding local condition fails, and may only
declare $\errC$ otherwise.

%
%
\begin{apxlemmarep}
\label{prop:Buchi_win-cond}
$W_\M$ is recognised by an $\NTA_1$.
\end{apxlemmarep}

\begin{proof}
	\label{app:Buchi_win-cond}
	As mentioned earlier, we need to prove that every timed language in the function $\projA^{-1}$
	of the union defining $V_\M^\Time$ can be expressed by an $\NTA_1$.
	
	First, we define the notion of local language.
	Given $\Lang$ a timed language of infinite words, its local version denoted by $\Lang^\ell$
	is the timed language of finite words such that:
	\begin{itemize}
		\item $\widehat{\Lang} = \widehat{\Lang^\ell}
		\cdot \Time(\Sigma_{\mathcal{M}}^\omega)$;
		\item $w \in \Lang$ if and only if all prefixes of $w$
		belong to $\Lang^\ell$.
	\end{itemize}
	Intuitively, local versions of $\RegIMT$, $A_{\mathcal{M}}^{\Time}$,
	$B_{\mathcal{M}}^{\Time}$, and $C_{\mathcal{M}}^{\Time}$ recognise the absence of errors of types
	${\RegI}$, ${A}$, ${B}$, and ${C}$.
	
	Now, we prove that the winning condition can be expressed by an $\NTA_1$.
	By definition of $\RegIMTloc$, this timed language is defined
	as a regular language over its untimed word.
	Thus, it can be expressed by a finite automaton.
	Now, we give the proof that $A_{\mathcal{M}}^{\Time, \ell}$ and $B_{\mathcal{M}}^{\Time, \ell}$
	can be expressed by a $\DTA_1$ by their complementary.
	
	\begin{lemma}
	    $\widehat{A_{\mathcal{M}}^{\Time, \ell}} \in \DTA_1$ and
	    $\widehat{B_{\mathcal{M}}^{\Time, \ell}} \in \DTA_1$.
	\end{lemma}
	\begin{nestedproof}
	    We start with $\widehat{A_{\mathcal{M}}^{\Time, \ell}}$.
	    By definition, $w \in \widehat{A_{\mathcal{M}}^{\Time, \ell}}$ when
	    the last two consecutive symbols of $w$ appear with the same timestamp.
	    Thus, the $\DTA_1$ that recognises $\widehat{A_{\mathcal{M}}^{\Time, \ell}}$
	    is depicted in the left of \cref{fig:buchi_win-cond-AB}
	    (a transition labelled with $\Sigma_\M$ means one transition for each letter in $\Sigma_\M$,
	    with the same guard and reset set of clocks).
	
	    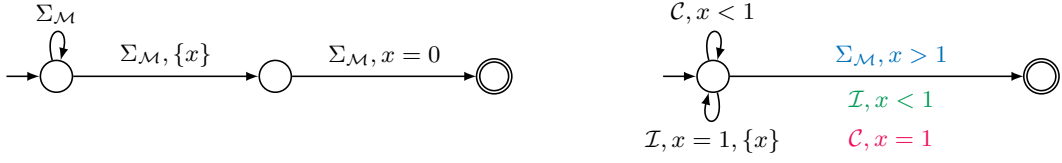
\begin{figure}
	        \centering
	        \begin{tikzpicture}[
    semithick,
    every state/.style={minimum size=12pt,inner sep=0},
    every node/.append style={font=\small},
    yscale=1,
    xscale=0.96,
    initial text={},
    every initial by arrow/.append style={inner sep=0, outer sep=0}
]
    \node[state,initial] at (0, 0)  (s0) {};
    \node[state] at (3, 0) (s1) {};
    \node[state, accepting] at (6, 0) (s2) {};

    \draw[->]
    (s0) edge node[above] {$\Sigma_\M, \{x\}$} (s1)
    (s0) edge[loop above] node[above] {$\Sigma_\M$} (s0)
    (s1) edge node[above] {$\Sigma_\M, x = 0$} (s2)
    ;

    \begin{scope}[xshift=5.5cm]
        \node[state,initial] at (3.5, 0) (s1) {};
        \node[state,accepting] at (8, 0) (s2) {};

        \draw[->]
        (s1) edge[loop above] node[above] {$\Counter, x < 1$} (s1)
        (s1) edge[loop below] node[below] {$\Instructions, x = 1, \{x\}$} (s1);

        \draw[->]
        (s1) -- node[pos=0.55,below] {$\begin{aligned}
            &\color{ForestGreen}\Instructions, x < 1\\
            &\color{OrangeRed}\Counter, x = 1
        \end{aligned}$} node[pos=0.55,above] {$\color{RoyalBlue}\Sigma_\M, x > 1$} (s2);
    \end{scope}
\end{tikzpicture}
	        \caption{The two $\DTA_1$ recognising $\widehat{A_{\mathcal{M}}^{\Time, \ell}}$ (on the left),
	            and $\widehat{B_{\mathcal{M}}^{\Time, \ell}}$ (on the right).}
	        \label{fig:buchi_win-cond-AB}
	    \end{figure}
	
	    Finally, we prove that $\widehat{B_{\mathcal{M}}^{\Time, \ell}} \in \DTA_1$.
	    By definition, if $w \in \widehat{B_{\mathcal{M}}^{\Time, \ell}}$ several cases may happen
	    depending on the last letter of $w$ denoted $(a, t)$:
	    \begin{itemize}
		\item $a$ is the first instruction, and $t \neq 1$: the first instruction does not start at $1$;
	
		\item $a \in \Instructions$, and $t \notin \N$: an instruction appears on a non-integer timestamps
		(depicted in green in \cref{fig:buchi_win-cond-AB});

		\item $a \in \Counter$ and $n \in \N$: another symbol appears with an integer timestamp
		(depicted by the orange edge in \cref{fig:buchi_win-cond-AB});

		\item by letting $w = w' \cdot (a', t') \cdot (a, t)$ with $t' < n< t$ for some $n \in \N$:
		no symbol appears in an integer timestamp (depicted by the blue edge in \cref{fig:buchi_win-cond-AB}).
	    \end{itemize}
	    Thus, the $\DTA_1$ that recognises $\widehat{B_{\mathcal{M}}^{\Time, \ell}}$
	    is depicted in the right of Figure~\ref{fig:buchi_win-cond-AB}.
	\end{nestedproof}

	\begin{corollary}
	    $A_{\mathcal{M}}^{\Time, \ell} \in \NTA_1$ and
	    $B_{\mathcal{M}}^{\Time, \ell} \in \NTA_1$.
	\end{corollary}

	To conclude the proof, we have to prove that knowing if a (finite\footnote{If such an error occurs, it would be on
	a prefix of a timed word played by \Timer.})
	timed word satisfies $C_{\mathcal{M}}^{\Time, \ell}$ can be done with an $\NTA_1$
	when we suppose that the words is in $\RegIMTloc \cap
	\AMloc \cap \BMloc$ (since otherwise the language $C_{\M}^{\Time, \ell}$ is not defined).

	\begin{lemma}
	    \label{lem:Buchi_win-cond_local}
	    $C^{\Time, \ell}_\M \cap \RegIMTloc \cap \AMloc \cap \BMloc \in \NTA_1$.
	\end{lemma}
	\begin{nestedproof}
		Let $w$ be a timed word.
		We want to test when $w \in C^{\Time, \ell}_\M \cap \RegIMTloc \cap \AMloc \cap \BMloc$ with a $\NTA_1$.
		Since we can know when $w$ is in $\RegIMTloc$, $\AMloc$ and $\BMloc$ with a $\DTA_1$,
		we suppose without generality that $w \in \RegIMTloc \cap \AMloc \cap \BMloc$, and we want to
		check whereas $w \in C^{\Time, \ell}_\M$ with a $\NTA_1$.

	    The key here is to remark that we only test the validity of the last letter of $w$
	    according to the last instruction of $w$.
	    As we use a lossy counter machine, we do not need to test if a
	    letter is missing (or not) when the last letter of $w$ is played.
	    Otherwise, one clock would not be enough.
	
		Let $w \in \RegIMTloc \cap \AMloc \cap \BMloc$,
		we want to check if it belongs to $C_{\M}^{\Time, \ell}$ with an $\NTA_1$.
		Since we assumed to be in $\BMloc$, $C_{\M}^{\Time, \ell}$ is trivially true if the last letter of $w$ is in
		$\Instructions$, so in that case we simply accept $w$.
		Otherwise, since in particular $w \in \RegIMTloc$, we have that
		\begin{displaymath}
		\untimed(w) = \enc(\val_0)\, \instruction_1\, \cdots\, \enc(\val)\, \instruction\, \enc(\val')
		\end{displaymath}
		with $\val$ being the complete valuation before the last instruction $\instruction$
		and $\val'$ being the possibly incomplete valuation after executing $\instruction$.
		Let $\counter \in \Counter$ the counter such that $\instruction \in \Instructions^\counter$.
		We distinguish three cases depending on the type of operation of $\instruction$, and build one corresponding $\NTA_1$ for each case.
	
	    \paragraph*{\boldmath If $\instruction \in \Instructions_\zt^\counter$:}
	    Note that the operation of $\instruction$ does not change the valuation,
	    i.e., we just have $\val' \geq \val$ due to the lossy semantics.
	    Thus, for $w$ to satisfy the local version of $C_{\mathcal{M}}^{\Time}$,
	    we just need to check that an occurrence of the last letter of $w$ exists in $\enc(\val)$
	    with the same fractional part.
	    Formally, $w$ must satisfy the following formula:
	    \begin{equation}
	        \tag{A}
	        \label{eq:Buchi_win-cond_local-0test}
	        w = v \cdot (a, t) ~\wedge~(a, t -1) \in w
	    \end{equation}
	    In particular, this can be tested by the $\NTA_1$ with the red branch depicted in \cref{fig:Buchi_win-cond_local}.

	    \paragraph*{\boldmath If $\instruction \in \Instructions_{\text{inc}}^\counter$:}
	    Observe that the operation of $\instruction$ may increase the valuation of the counter $\counter$
	    (the valuation of other counters remains stable or decreases).
	    In particular, when the first occurrence of $\counter$ is fresh, we need to test
	    if it appears within one unit of time after the first occurrence of $\counter$ in $\enc(\val)$ (if it exists).
		Otherwise, we need to test the presence of an occurrence of the last letter one unit of time before
		(as for the previous case).
	
		Let $(a, t)$ be the last letter in $w$.
		In this case, $w$ belongs to $C_{\M}^{\Time, \ell}$ 
		when it satisfies the formula $A \vee B \vee C \vee D$ as defined below.
	
	    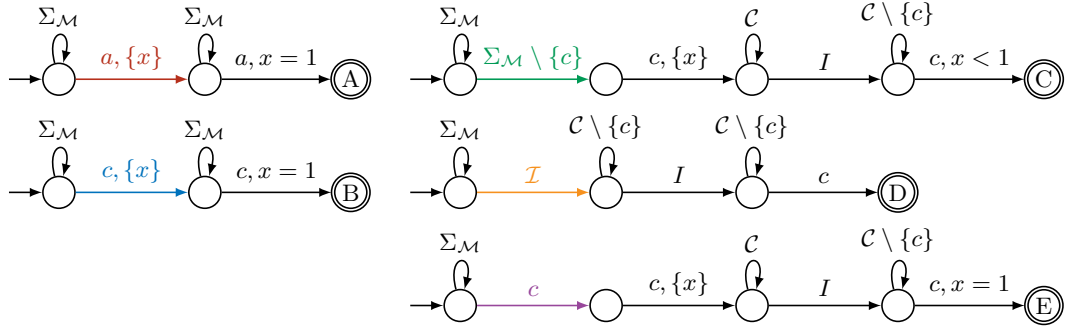
\begin{figure}
	        \centering
	        \begin{tikzpicture}[
    semithick,
    every state/.style={minimum size=12pt,inner sep=0},
    every node/.append style={font=\small},
    yscale=1,
    xscale=0.96,
    initial text={},
    every initial by arrow/.append style={inner sep=0, outer sep=0}
]
    \begin{scope}
        \node[state, initial] at (0, 0)  (s0) {};
        \node[state] at (2, 0) (s2) {};
        \node[PlayerMin, accepting] at (4, 0) (s4) {\ref{eq:Buchi_win-cond_local-0test}};

        \draw[->]
        (s0) edge[loop above] node[above] {$\Sigma_\M$} (s0)
        (s0) edge[BrickRed] node[above] {$a, \{x\}$} (s2)
        (s2) edge[loop above] node[above] {$\Sigma_\M$} (s2)
        (s2) edge node[above] {$a, x = 1$} (s4);
    \end{scope}

    \begin{scope}[xshift=5.5cm,yshift=-1.5cm]
        \node[state, initial] at (0, 0)  (s0) {};
        \node[state] at (2, 0) (s2) {};
        \node[state] at (4, 0) (s3) {};
        \node[PlayerMin, accepting] at (6, 0) (s4) {\ref{eq:Buchi_win-cond_local-inc-D}};

        \draw[->]
        (s0) edge[loop above] node[above] {$\Sigma_\M$} (s0)
        (s0) edge[BurntOrange] node[above] {$\Instructions$} (s2)
        (s2) edge[loop above] node[above] {$\Counter \setminus \{\counter\}$} (s2)
        (s2) edge node[above] {$\instruction$} (s3)
        (s3) edge[loop above] node[above] {$\Counter \setminus \{\counter\}$} (s3)
        (s3) edge node[above] {$c$} (s4)
        ;
    \end{scope}

    \begin{scope}[yshift=-1.5cm]
        \node[state, initial] at (0, 0)  (s0) {};
        \node[state] at (2, 0) (s2) {};
        \node[PlayerMin, accepting] at (4, 0) (s4) {\ref{eq:Buchi_win-cond_local-inc-B}};

        \draw[->]
        (s0) edge[loop above] node[above] {$\Sigma_\M$} (s0)
        (s0) edge[RoyalBlue] node[above] {$c, \{x\}$} (s2)
        (s2) edge[loop above] node[above] {$\Sigma_\M$} (s2)
        (s2) edge node[above] {$c, x = 1$} (s4)
        ;
    \end{scope}

    \begin{scope}[yshift=0cm, xshift=5.5cm]
        \node[state, initial] at (0, 0)  (s0) {};
        \node[state] at (2, 0) (s1) {};
        \node[state] at (4, 0) (s2) {};
        \node[state] at (6, 0) (s3) {};
        \node[PlayerMin, accepting] at (8,0) (s4) {\ref{eq:Buchi_win-cond_local-inc-C}};

        \draw[->]
        (s0) edge[loop above] node[above] {$\Sigma_\M$} (s0)
        (s0) edge[ForestGreen] node[above] {$\Sigma_\M \setminus \{\counter\}$} (s1)
        (s1) edge node[above] {$c, \{x\}$} (s2)
        (s2) edge[loop above] node[above] {$\Counter$} (s2)
        (s2) edge node[above] {$\instruction$} (s3)
        (s3) edge[loop above] node[above] {$\Counter \setminus \{\counter\}$} (s3)
        (s3) edge node[above] {$c, x < 1$} (s4);
    \end{scope}

    \begin{scope}[yshift=-3cm, xshift=5.5cm]
        \node[state, initial] at (0, 0)  (s0) {};
        \node[state] at (2, 0) (s1) {};
        \node[state] at (4, 0) (s2) {};
        \node[state] at (6, 0) (s4) {};
        \node[PlayerMin, accepting] at (8, 0) (s5) {\ref{eq:Buchi_win-cond_local-dec-E}};

        \draw[->]
        (s0) edge[loop above] node[above] {$\Sigma_\M$} (s0)
        (s0) edge[Purple] node[above] {$\counter$} (s1)
        (s1) edge node[above] {$\counter, \{x\}$} (s2)
        (s2) edge[loop above] node[above] {$\Counter$} (s2)
        (s2) edge node[above] {$\instruction$} (s4)
        (s4) edge[loop above] node[above] {$\Counter \setminus \{\counter\}$} (s4)
        (s4) edge node[above] {$\counter, x = 1$} (s5);
    \end{scope}
\end{tikzpicture}
	        \caption{The gadgets used to recognise $C_\mathcal{M}^{\Time, \ell}$
				by using the formula~\eqref{eq:Buchi_win-cond_local-0test}, ~\eqref{eq:Buchi_win-cond_local-inc-B},
				~\eqref{eq:Buchi_win-cond_local-inc-C},~\eqref{eq:Buchi_win-cond_local-inc-D},
				~\eqref{eq:Buchi_win-cond_local-dec-E}.
			}
	        \label{fig:Buchi_win-cond_local}
	    \end{figure}

	    \subparagraph*{\boldmath Case $A$: $a$ is not an occurrence of $\counter$.}
	    We suppose that $a \neq \counter$.
	    Intuitively, we are in the same case that for the zero-test since the operation of $\instruction$
	    does not increase the value of the counter $a$.
	    Thus, we need to check if $w$ satisfies~\eqref{eq:Buchi_win-cond_local-0test}, which can be tested
	    by the $\NTA_1$ with the red branch depicted in \cref{fig:Buchi_win-cond_local}.
	
	    \subparagraph*{\boldmath Case $B$: $a$ is a repeated occurrence of $\counter$.}
	    This is the case where $a = \counter$ is not fresh, i.e., $(\counter, t - 1) \in w$.
	    Formally, $w$ has the form
	    \begin{equation}
	        \tag{B}
	        \label{eq:Buchi_win-cond_local-inc-B}
	        w = u \, (\counter, t-1)\, v\, (\counter, t)
	    \end{equation}
	    that can be tested by the $\NTA_1$ with the blue branch depicted in \cref{fig:Buchi_win-cond_local}.
	
	    \subparagraph*{\boldmath Case $C$: $a$ is a fresh occurrence of $\counter$ and $\val(c) \geq 1$.}
		This is the case where $a = \counter$ is a fresh first occurrence of $\counter$ after $\instruction$,
		and the incrementation is made on a counter with a positive valuation (i.e., $\counter \in \enc(\val)$).
		In particular, $a$ is the first occurrence of $\counter$ after $\instruction$ within one unit of time after
		the first one in $\enc(\val)$.
		Formally, $w$ satisfies the following formula:
		\begin{equation}
			\tag{C}
			\label{eq:Buchi_win-cond_local-inc-C}
			w= v_1 \cdot (\instruction, t_\instruction) \cdot v_2 \cdot (\counter, t) ~\wedge~
			c \notin \untimed(v_2)  ~\wedge~
			\exists (\sigma, t_1) \cdot (\counter, t_2) \subseteq v_1.
			(t < t_2 + 1 \wedge \sigma \neq c)
		\end{equation}
		that can be tested by the $\NTA_1$ with the green branch depicted in \cref{fig:Buchi_win-cond_local}.
	
	    \subparagraph*{\boldmath Case $D$: $a$ is a fresh occurrence of $\counter$ and $\val(c) = 0$.}
	    This is the case where $a = \counter$ is the first occurrence of $\counter$ after $\instruction$,
		and the increment of $\counter$ was made for the first time (i.e., $\counter \notin \enc(\val)$).
		Here, we need to test that no occurrence of $\counter$ appears in $\enc(\val)$,
		nor after $\instruction$ until the last letter.
		Moreover, since we assumed $w \in \RegIMTloc$, we know this $\counter$ will be correctly placed
		(for instance, $\counter = \counter_0$ will be between an instruction and the $\counter_1$ sequence).
		Formally, $w$ satisfies the following formula:
		\begin{equation}
			\tag{D}
			\label{eq:Buchi_win-cond_local-inc-D}
			w= v_1 \cdot \Instructions \cdot v_2 \cdot (\instruction, t_\instruction) \cdot v_3 \cdot (\counter, t)
			~\wedge~ c \notin \untimed(v_2)  ~\wedge~ c \notin \untimed(v_3)
		\end{equation}
		that can be tested by the $\NTA_1$ with the orange branch depicted in \cref{fig:Buchi_win-cond_local}.

	    \paragraph*{\boldmath If $\instruction \in \Instructions_{\text{dec}}^\counter$:}
	    When the operation of $\instruction$ decrements the counter $\counter$, we need to check that the
	    first occurrence of $\counter$ in $\enc(\val)$ is ``deleted'', i.e., does not reappear one time unit later in $\enc(\val')$.
	    Again, since $w \in \RegIMTloc \cap \AMloc$, this occurrence exists
	    and it is unique.
	    Formally, let $(a, t)$ be the last letter in $w$.
	    $w$ belongs to $C_{\M}^{\Time, \ell}$
	    if it satisfies the formula $A \vee E$
	    as defined below.
	
		\subparagraph*{\boldmath Case $A$: $a$ is not an occurrence of $\counter$.}
		If $a \neq \counter$, we are again in the same case as for the zero-test since the operation of $\instruction$
		does not increase the value of the counter $a$.
		Thus, we need to check if $w$ satisfies~\eqref{eq:Buchi_win-cond_local-0test} that can be tested
		by the $\NTA_1$ with the red branch depicted in \cref{fig:Buchi_win-cond_local}.

		\subparagraph*{\boldmath Case $E$: $a$ is an occurrence of $\counter$.}
		If $a = \counter$, it can not be fresh (as no incrementation is possible), i.e., $(a, t - 1) \in w$.
		Moreover, $(\counter, t - 1)$ must not be the first occurrence of $\counter$ in $\enc(\val)$.
		Formally, $w$ satisfies the following formula:
		\begin{equation}
			\tag{E}
			\label{eq:Buchi_win-cond_local-dec-E}
			w= v \cdot (\counter, t) ~\wedge~ \exists (\counter, \tau_1) \colon
			(\counter, \tau_1) \cdot (\counter, \tau-1) \subseteq w
		\end{equation}
		that can be tested by the $\NTA_1$ with the purple branch depicted in \cref{fig:Buchi_win-cond_local}.

	Finally, we build the $\NTA_1$ for $C_{\M}^{\Time, \ell}$ by first guessing which of the three cases the last instruction $\instruction$ satisfy, and then running the union of the appropriate $\NTA_1$ to check if the word satisfy the corresponding formula.
	\end{nestedproof}
	We have shown that $A_{\M}^{\Time, \ell}$, $B_{\M}^{\Time, \ell}$,  and $C_{\M}^{\Time, \ell}$ are recognised by $\NTA_1$.
	Therefore, so is $W_\M$.
	This ends the proof of Lemma~\ref{prop:Buchi_win-cond}.
\end{proof}

\para{Correctness of reduction}
We prove the following lemma:
\begin{restatable}{lemma}{lemBuchiCorrectness}
    \label{lem:Buchi_correctness}
    \Timer\ has a winning controller if $\M$ repeatedly reaches $\state$, and
    \Monitor has a winning strategy otherwise.
\end{restatable}
\begin{proof} 
    Due to a well-quasi-order on configurations
    and the lossy semantics, if $\M$ repeatedly reaches $\state$ then $\M$ has a \emph{lasso}
    run  repeatedly reaching $\state$~\cite{Schnoebelen10a}: a run that is cyclic from some point on.
    The run is thus bounded by some $k\in\N$, and \Timer has a winning controller that produces
    an encoding of this run of granularity at most $\delta=\nicefrac{1}{4(k+1)}$.
	Indeed, with this granularity, \Timer can encode all potential valuations by
	using $k$ ``slots'' in an interval of length $\nicefrac{1}{4}$ for each of the segments.

    Conversely, if $\M$ does not repeatedly reach $\state$ then whenever
    \Timer\ produces a correct run it necessarily visits $\state$ finitely often.
    Then \Monitor has a winning strategy that  records all the history
    and keeps playing $\Ok$ as long as \Timer is not cheating,
    and is able to detect all kinds of errors produced by \Timer in case
    \Timer\ cheats.
\end{proof}

\section{\boldmath\Timer is the \owner but not the \subject}
\label{sec:boundedness}

In this section we prove undecidability in the last remaining case:

\begin{restatable}{theorem}{thmLCMboundedness}
    \label{thm:Radek}
    The timed Church synthesis problem is undecidable when $\Owner = \Timer$ and $\Subject = \Monitor$.
\end{restatable}

\noindent
(Note the similarity of \cref{thm:Radek,thm:universality-sample} along the exchange of roles of \Timer and \Monitor.)

We reduce from the LCM boundedness problem, assuming free-test semantics.
Let us fix a 4-counter LCM $\M = \tuple{\Counter,\States,\state_0,\Instructions}$,
where $\Counter=\{c_1,c_2,c_3,c_4\}$.
Let $\TimerAlphabet = \Instructions \cup \{\pool\}$ and $\MonitorAlphabet = \{\ok,\err\}$.
We define a timed game with \Timer's winning condition $W_\M$ such that
\Monitor has a winning controller if and only if $\M$ is bounded.

\para{\boldmath 1-resetting $\NTA_1$}
In this case, we need a slight extension of $\NTA_1$ to define winning conditions and controllers:
\NTA with 1 clock and very limited form of $\varepsilon$-transitions that reset the clock
every time it equals 1.
Let $\FracPart = \Qpos \cap [0,1)$ denote the set of fractional parts.
For any time value $t \in \R_{\geq 0}$, let $\fracpart{t}$
denote its fractional part, i.e., the unique value in $\FracPart$ such that
$t = n + \fracpart{t}$ for some $n \in \N$.
If $\val$ is a valuation, $\fracpart{\val}$ denotes the valuation $\{x \mapsto \fracpart{\val(x)}\}$.

A \emph{1-resetting $\NTA_1$} ($\resNTA_1$) is an $\NTA_1$
as defined in Section~\ref{sec:prelim} with only the following modification:
$(\loc, \val) \xrightarrow{\delta, \tau} (\loc', \val')$ (with $a \in \Sigma$) if and only if
$\delta = (\loc, a, g, Y, \loc')\in \Trans$ is a transition of $\A$
such that $\fracpart{\val + \tau} \models g$, and $\val' = \fracpart{\val + \tau}[Y \coloneqq 0]$.
The 1-resetting $\NTA_1$ can be simulated using $\NTA_1$ with $\varepsilon$-transitions:
in every location $\loc \in \Locs$ add a self-loop $(\loc,\varepsilon,x=1,\{x\},\loc)$.

%
%
%
%
%

\para{The idea of reduction}
Intuitively speaking,
\Timer is tasked with simulating a run of $\M$, and \Monitor can point out when they think that \Timer made a mistake in the simulation.
More specifically, \Timer will play instructions of $\M$, and the time values will be used to encode the valuations of counters, as described below
(the encoding is different from the one in the previous section).
After each move by \Timer, \Monitor can either say that the simulation is correct so far, or that \Timer made a mistake in their last move.
In the former case, if \Timer actually made a mistake then \Timer wins
(\Monitor must point out any mistake), and
if there was no mistake then the game continues
(if the game continues like this forever, \Monitor wins).
In the latter case, the game is essentially immediately over: 
either \Monitor is right and \Timer made a mistake, which will be winning for \Monitor, or \Monitor is wrong and \Timer's last move was correct, in which case \Timer will be winning.

The way time is used to encode counter valuations is the following.
The fractional part of a timestamp acts as an identifier which allows an increment to be later matched with a corresponding decrement.
For instance, an increment at time 2.314 can be later matched with a decrement at time 7.314.
With this, the valuation of a counter is simply the number of increments/fractional parts that have not been matched with a later decrement with the same fractional part.
See Figure~\ref{figure:run encoding} for an illustration.
%
\begin{figure}
\begin{center}\begin{tikzpicture}[xscale=1.89]
    \def\tstart{0}
    \def\tend{7}

    \foreach \x in {\tstart,0.1,...,\tend} {
        \draw[black!30, line width=0.3pt] (\x,0.12) -- (\x,-0.12);
    }

    \foreach \x in {\tstart,1,...,\tend} {
        \draw[black, line width=0.6pt] (\x,0.12) -- (\x,-0.12);
        \node[below=2pt, font=\small] at (\x,-0.12) {\x};
    }

    \draw[thick,->] (\tstart-0.2,0) -- (\tend+0.2,0);

    \def\events{%
            {1/0.1/$\square$/c1},%
            {2/0.4/$\square$/c5},%
            {3/0.6/$\square$/c4},%
            {4/0.75/$\square$/c3},%
            {1b/1.1/$\square$/c1},%
            {5/1.3/$\square$/c2},%
            {8/2.3/$c_1\mkern-2mu\increment$/c2},%
            {6/2.6/$c_1\mkern-2mu\increment$/c4},%
            {7/3.6/$c_1\mkern-2mu\decrement$/c4},%
            {dd/4.4/$c_2\mkern-2mu\increment$/c5},%
            {10/4.75/$c_1\mkern-2mu\increment$/c3},%
            {9/5.3/$c_1\mkern-2mu\decrement$/c2},%
            {ee/3.1/$c_4\mkern-2mu\increment$/c1},%
            {11/5.75/$c_1\mkern-2mu\decrement$/c3},%
            {12/6.3/$c_1\mkern-2mu\increment$/c2}%
    }
    \coordinate (lfin) at (7.2,0.2) {};
    \coordinate (bfin) at (7.2,-0.2) {};
    \coordinate (bini) at (1.7,-0.2) {};
    \foreach \id/\x/\label/\c in \events {
        \coordinate (e\id) at (\x,0.2) {};
        \coordinate (b\id) at (\x,-0.2) {};
        \fill[fill=\c] (e\id) -- ++(0.03,0.2) -- ++(-0.06,0) -- cycle;
        \fill[black] (e\id) -- ++(0.0225,0.15) -- ++(-0.045,0) -- cycle;
        \node[font=\scriptsize, above=3mm of e\id,inner sep=0] (l\id) {\label};
    }

    \draw[overbrace style] (l1.north west) -- node[overbrace text style] (name) {pool} (l5.north east);

    \newcommand{\drawArc}[6]{
        \coordinate (p1) at ($(l#1.north)+(0,1mm)$);
        \coordinate (p2) at ($(l#1.north)+(0,#4)$);
        \coordinate (p3) at ($(l#2.north)+(0,#4)$);
        \coordinate (p4) at ($(l#2.north)+(0,1mm)$);

        \draw[#3,rounded corners=2mm] (p1) -- (p2) --node[pos=#6,color=#3!70!black,above, font=\small]{#5} (p3) -- (p4);
        \draw[#3] (b#1) -- ++(0,-1.1);
        \draw[#3] (b#2) -- ++(0,-1.1);
    }
    \newcommand{\drawArcIncomplete}[6]{
        \coordinate (p1) at ($(l#1.north)+(0,1mm)$);
        \coordinate (p2) at ($(l#1.north)+(0,#4)$);
        \coordinate (p3) at ($(l#2 |- l#1.north)+(0,#4)$);

        \draw[#3,rounded corners=2mm] (p1) -- (p2) --node[pos=#6,color=#3!70!black,above, font=\small]{#5} (p3);
        \draw[#3] (b#1) -- ++(0,-1.1);
    }
    \drawArc{6}{7}{c4}{0.8cm}{\textvisiblespace .6 active}{0.5}
    \drawArc{8}{9}{c2}{0.4cm}{\textvisiblespace .3 active}{0.62}
    \drawArc{10}{11}{c3}{0.8cm}{\textvisiblespace .75 active}{0.5}
    \drawArcIncomplete{12}{fin}{c2}{0.4cm}{\textvisiblespace .3 active}{0.55}

    \node[anchor=west,inner sep=0] at (-0.04,-0.95) {simulated value of $c_1$:};

    \foreach \x/\y/\val in {{ini/8/0},{8/6/1},{6/7/2},{7/10/1},{10/9/2},{9/11/1},{11/12/0},{12/fin/1}} {
        \coordinate (bl) at ($(b\x)+(0,-10mm)$);
        \coordinate (br) at ($(b\y)+(0,-10mm)$);
        \draw[<->,dotted] (bl) -- node[midway,above] {\val} (br);
    }

\end{tikzpicture}\end{center}
\caption{Example encoding of a run. For each instruction $(s, \op, s')$, only $\op$ is shown.}
\label{figure:run encoding}
\end{figure}
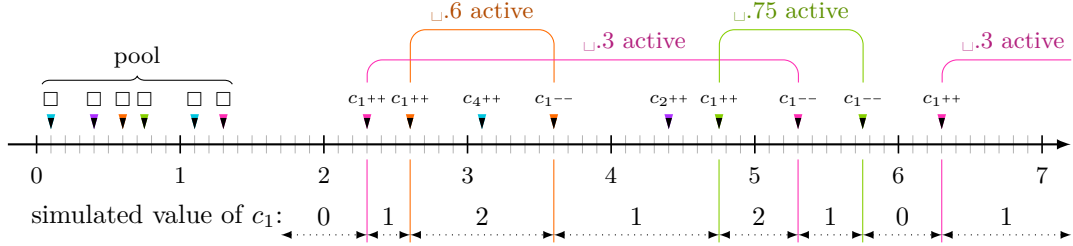

An obvious implementation of this idea would require,
in the case where a fractional part is used for the first time ever in an incrementing instruction,
checking that it is fresh using unrestricted ``guessing'' $\varepsilon$-transition
(i.e., guess the fractional part before the first instruction, then check it never appears until the very last position).
To avoid this, we introduce a pool of fractional parts at the start of the run encoding, denoted by the symbol~``\pool''.
The idea is that \Timer must include at the beginning a number of relevant fractional parts that will be used later in the run encoding.
Then, we can eliminate time guessing by instead non-deterministically choosing 
a \emph{position} in the pool, resetting the clock at that point and then every time it equals 1.
If the last position in the word is the only position seen with the clock exactly at zero, then we know this fractional part was never used before.

\para{\boldmath The encoding of runs of $\M$}
We define a timed word encoding of runs of $\M$ in three parts.
First, we capture the regular properties of their projection onto~$\TimerAlphabet$.
Let $\projAA \colon (\TimerAlphabet \times \MonitorAlphabet \times \Qpos)^* \to \TimerAlphabet^*$ be the natural projection.
Define $\RegII\subseteq (\TimerAlphabet \times \MonitorAlphabet \times \Qpos)^*$ as the set of all finite timed words $w$ such that
$\projAA(w) = \pool^k \instruction_1 \instruction_2 \cdots \instruction_n$,
instruction $\instruction_1$ starts from the initial location, and each $\instruction_i$ is compatible with $\instruction_{i+1}$ in the sense defined earlier.

It remains to state the role of timestamps, on positions with a symbol in $\Instructions$,
in maintaining the valuations of $\M$.
Consider a finite timed word $w \in (\TimerAlphabet \times \MonitorAlphabet \times \Qpos)^*$.
We say that a fractional part $f \in \FracPart$ is \emph{active} for counter $c$ in $w$ if the last timed letter $(I, t)$ with $I \in\Instructions^c$ and $\fracpart t = f$ increments $c$ ($I \in \Instructions^c_\inc$) and appears after the last zero-test of $c$.
Conversely, $f$ is \emph{inactive} for $c$ in $w$ if since the last zero-test of $c$ there is no occurrence of $f$ with an instruction $\instruction \in \Instructions^c$ involving $c$ or if the last such occurrence is in $\Instructions^c_\dec$.
For any prefix $w$ of a run encoding, the corresponding valuation of counter $c$ is the number of distinct fractional parts that are active for $c$.
Let $\wordToVal(w)$ denote this valuation.

Recall the form of runs of $\M$:
$\run = (\state_0, \val_0) \xrightarrow{\instruction_1} \, (\state_1, \val_1) \xrightarrow{\instruction_2}
\cdots$, as an alternating sequence of configurations and instructions, with
$\val_0 = \zeroval$.
We inductively define $\run(w)$, a finite alternating sequence of configurations and instructions
associated with any $w \in \RegII$:
\begin{align*}
\run(w) = 
\begin{cases}
(\state_0, \val_0) &\text{ if } \projAA(w)\in \pool^\ast,\\
\run(w') \xrightarrow{\instruction}  (\state,\wordToVal(w)) &\text{ if } \projAA(w) = w' \cdot \instruction
\text{ and } \instruction = (\state',\op,\state).
\end{cases}
\end{align*}
Note that $\run(w)$ need not be a run of $\M$ for two reasons.
First, valuations need not be correlated with the instructions (e.g., a run may feature an incrementing instruction while the corresponding valuation remains unchanged or decreases).
Second, a decrementing instruction may occur even when the valuation of the corresponding counter is zero.
Conversely, if valuations are updated correctly and decrementing instructions occur only when the valuation is non-zero, then $\RegII$ ensures that $\run(w)$ is a valid run of $\M$.
To enforce validity, we introduce rules that must be satisfied by the timestamp $t$
of every instruction $\instruction$ occurring in $w$:

\begin{description}
\item[Rule 1:]
If $\instruction \in \Instructions^c_\inc$, then
the fractional part of $t$ must be inactive for $c$.
\item[Rule 2:]
If $\instruction \in \Instructions^c_\dec$, then
the fractional part of $t$ must be active for $c$.
\end{description}

It is easy to see that following both rules guarantees maintaining the correct counter valuations.
Moreover, it prevents one from adding a decrementing instruction when the counter valuation is zero, as there would be no active fractional part to pick.
Note that there are no constraints on zero-test instructions, their timestamps are irrelevant.

\begin{claim}
    \label{rem:correctness-of-active-based-encoding}
    If all instructions in $w \in \RegII$ satisfy Rules 1 and 2, then $\run(w)$ is a run of $\M$.
\end{claim}

\para{The timed game}

For each of the two rules, we now define a language of finite timed words
that requires that the rule is satisfied by the last letter, and another one that requires that
the rule is broken.
Let $\ErrLangOne$ be the language of finite timed words where the last letter is an increment instruction
that breaks Rule 1.
Moreover, let
$\NotErrLangOne$ be  the language of words where the last letter satisfies Rule 1
and
\emph{the last fractional part is in the pool}.
Note that any timed word ending in either a decrementing or zero-test instruction is in $\NotErrLangOne$, and is not in $\ErrLangOne$.
$\ErrLangOne$ is recognised by the 1-resetting reachability $\NTA_1$ given in Figure~\ref{figure:NTA Err1}, while
$\NotErrLangOne$
is recognised by the 1-resetting reachability $\NTA_1$ given in Figure~\ref{figure:NTA NotErr1}.

\begin{figure}
\begin{center}
    \begin{tikzpicture}[
    semithick,
    every state/.style={minimum size=12pt,inner sep=0},
    every node/.append style={font=\small},
    yscale=1,
    xscale=0.96,
    initial text={},
    every initial by arrow/.append style={inner sep=0, outer sep=0}
]
    \node[state,initial] (2) at (4.5,0) {};
    \node[state] (3) at (9,0) {};
    \node[state,accepting] (4) at (13.65,0) {};

    \path[->] (2) edge [loop above] node [left] {$\TimerAlphabet$} (2);
    \path[->] (2) edge node [above] {$\Instructions^c_\inc, \set{x}$} (3);
    \path[->] (3) edge [loop above] node [above] {$
    \begin{aligned}
    \TimerAlphabet \setminus \Instructions^c&\\[-2pt]
    \Instructions^c_\inc \cup \Instructions^c_\dec&, x>0
    \end{aligned}$} (3);

    \path[->] (3) edge node [above] {$\Instructions^c_\inc, x=0$} (4);
\end{tikzpicture}
\end{center}

\caption{An $\resNTA_1$ for $\ErrLangOne$ comprises four copies of the above automaton, one for each $c \in \Counter$.}
\label{figure:NTA Err1}
\end{figure}

The languages $\ErrLangOne$ and $\NotErrLangOne$ are readily seen to be disjoint.
Note however that the languages 
are not complements.
Indeed, a word that satisfies Rule 1 but whose last fractional part does not appear in the pool at the beginning belongs to neither language.

The definitions and automata for $\ErrLangTwo$ and $\NotErrLangTwo$ mirror those of $\NotErrLangOne$ and $\ErrLangOne$ respectively, with $\ErrLangTwo$ being the language of words breaking Rule 2 where 
\emph{the last fractional part is in the pool}, and $\NotErrLangTwo$ being the language of words where the last letter satisfies Rule 2.
Again, a word not ending in a decrement is automatically in $\NotErrLangTwo$ and not in $\ErrLangTwo$.

Let $\ErrLang = \RegII \cap \bigl(\ErrLangOne \cup \ErrLangTwo\bigr)$ and $\NotErrLang = \RegII \cap \NotErrLangOne \cap \NotErrLangTwo$.
We define \Timer's winning condition as the following language of infinite timed words:
\[
\Win = \Reach\left\{w \mid (\projB(w) \in \ok^\ast \err \land w \in \NotErrLang) \lor
(\projB(w) \in \ok^\ast \ok \land w \in \ErrLang)\right\}.
\]
%
Intuitively, \Timer wins if at any point \Monitor makes a mistake in its claim by playing $\err$ when the sequence given by \Timer is actually correct or playing $\ok$ when there is an error.
However the play continues after this point, \Timer will be winning.
On the other hand, \Monitor wins if either it plays $\err$ at a point in which the play is not in $\NotErrLang$, and then any continuation will be outside of $\Win$, or if it plays $\ok$ forever while the play is never in $\ErrLang$.
Note that the case where \Timer plays outside of $\RegII$ is covered both by not being in $\NotErrLang$,
so \Monitor can play $\err$ if this happens and win, or by
not being in $\ErrLang$,
so \Monitor can play $\ok$ if this happens and win as well.

\begin{lemma}
$\Win$ is recognised by a $\rresNTA_1$.
\end{lemma}
\begin{proof}
Both $\ErrLang$ and $\NotErrLang$ are recognisable by 1-resetting reachability $\NTA_{1}$s.
For $\ErrLang$ this follows from the closure of $\rresNTA_1$ under union of languages, and
intersection with $\RegII$ does not add clocks.
As for $\NotErrLang$, although intersection typically adds up the numbers of clocks,
we can branch based on the final letter to simulate either $\NotErrLangOne$ or $\NotErrLangTwo$, requiring only one clock.
\end{proof}
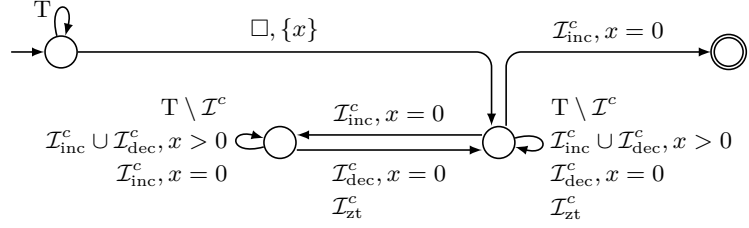
\begin{figure}
    \begin{center}\begin{tikzpicture}[
    semithick,
    every state/.style={minimum size=12pt,inner sep=0},
    every node/.append style={font=\small},
    yscale=1,
    xscale=0.96,
    initial text={},
    every initial by arrow/.append style={inner sep=0, outer sep=0}
]
    \node[state,initial] (2) at (-6,-2)  {};
    \node[state] (3) at (0,-3.2) {};
    \node[state] (4) at (-3,-3.2) {};
    \node[state,accepting] (5) at (3.15,-2) {};

    \path[->] (2) edge [loop above] node [left] {$\TimerAlphabet$} (2);

    \coordinate (corner1) at (-0.1,-2);
    \draw[->,rounded corners=2mm] (2) -- node[above] {$\pool, \set{x}$} (corner1) -- ($(3.north)+(-0.1,0)$);

    \path[->] (3) edge [loop right] node [right] {$
    \begin{aligned}
        &\ \\[-2.2pt]
        &\TimerAlphabet \setminus \Instructions^c\\[-2pt]
        &\Instructions^c_\inc \cup \Instructions^c_\dec, x>0\\[-2pt]
        &\Instructions^c_\dec, x=0\\[-2pt]
        &\Instructions^c_{\zt}
    \end{aligned}$} (3);

    \path[->,transform canvas={yshift=1mm}] (3) edge node[above] {$\Instructions^c_\inc,x=0$} (4);

    \path[->] (4) edge [loop left] node [left] {$
    \begin{aligned}
        \TimerAlphabet \setminus \Instructions^c\\[-2pt]
        \Instructions^c_\inc \cup \Instructions^c_\dec, x>0\\[-2pt]
        \Instructions^c_\inc, x=0\\[0.2pt]
    \end{aligned}$} (4);

    \path[->,transform canvas={yshift=-1mm}] (4) edge node [below] {$\begin{aligned}
        &\Instructions^c_\dec, x=0\\[-2pt]
        &\Instructions^c_{\zt}
    \end{aligned}$} (3);

    \coordinate (corner2) at (0.1,-2);
    \draw[->,rounded corners=2mm] ($(3.north) + (0.1,0)$) -- (corner2) -- node [above] {$\Instructions^c_\inc,x=0$} (5);
\end{tikzpicture}\end{center}

\caption{An $\resNTA_1$ recognising $\NotErrLangOne$ consists of four copies of the automaton above, one for each $c \in \Counter$, and an additional branch that checks that the last transition is in $\Instructions \setminus \Instructions_\inc$ (omitted).}
\label{figure:NTA NotErr1}
\end{figure}

\para{Correctness}

Before stating correctness of the reduction,
let us give some intuition about the pool at the beginning of the play.
It is in \Timer's best interest to accurately list as many fractional parts that will be used later, but only finitely many.
The finite part is obvious: if \Timer only plays $\pool$ forever then \Monitor just plays $\ok$ forever and wins, so \Timer must start the real run encoding at some point.
The incentive for \Timer to play many fractional parts comes from its winning condition: it wins if \Monitor makes a mistake.
For example, if \Timer plays an increment transition with some inactive fractional part $f$ and \Monitor answers with $\err$, \Timer wins iff the play is in $\NotErrLangOne$.
But $\NotErrLangOne$ accepts only if the last fractional part occurs in the pool.
So if $f$ is not in the pool, the play is not in $\NotErrLang$ (despite being a correct encoding of a run), and \Monitor wins by playing $\err$.
On the other hand, there is no possible disadvantage to adding more fractional parts to the pool, even some that will never be used later, and so this is \Timer's incentive for filling the pool as much as possible.


\begin{lemma}
    \label{lemma:correct encoding game}
Let $w\in\RegII$. We have the following implications:
\begin{enumerate}
\item $\run(w) \in \MRuns_\M \Rightarrow$ all prefixes $w'$ of $w$ satisfy $w' \notin \ErrLang$.
\item $\run(w) \notin \MRuns_\M \Rightarrow$ some prefix $w' = w'' \cdot  (\instruction,m,t)$ of $w$ satisfies $\run(w'') \in \MRuns_\M$, $\run(w') \notin \MRuns_\M$, and $w' \notin \NotErrLang$.
\end{enumerate}
\end{lemma}
\begin{proof}
\fbox{\textbf{1.}}
By induction on prefixes of $w$. The base case, for prefixes where \Timer has only played $\pool$ letters so far, is trivial.
Let $w'$ be a prefix of $w$.
Since $\MRuns_\M$ is prefix-closed and $\run(w')$ is a prefix of $\run(w)$, $\run(w') \in \MRuns_\M$.
Assume by induction hypothesis that all strict prefixes of $w'$ are not in $\ErrLang$.
Let $w' = w'' \cdot (\instruction,m,t)$ with $\instruction \in \Instructions$, $m \in \MonitorAlphabet$, and let $f = \fracpart{t} \in \FracPart$.
We look at the different cases for $\instruction$ to show that $w' \notin \ErrLang$.
\begin{itemize}
\item Suppose $\instruction \in \Instructions_\inc^c$.
We have $\run(w') \in \MRuns_\M$, $\run(w'') \in \MRuns_\M$, and $\run(w') = \run(w'') \xrightarrow{\instruction} (\state,\wordToVal(w'))$.
Thus $\wordToVal(w') = \wordToVal(w'')[c\increment]$, and from this we have that $f$ is inactive for $c$ in $w''$.
This does not mean that $w' \in \NotErrLangOne$ as $f$ could be missing from the pool, but it does mean that $w' \notin \ErrLangOne$ as $\ErrLangOne$ can only accept if the last fractional part is already active.
Thus $w' \notin \ErrLang$.

\item Suppose $\instruction \in \Instructions_\dec^c$.
Since the run is correct, we have that $\wordToVal(w'')(c) > 0$ and $\wordToVal(w') = \wordToVal(w'')[c\decrement]$, so $f$ must be active for $c$ in $w''$.
Then $w' \in \NotErrLangTwo$ which implies $w' \notin \ErrLang$.

\item If $\instruction \in \Instructions_\zt^c$ then $w'$ is trivially accepted by $\NotErrLangOne$ and $\NotErrLangTwo$ thus $w' \notin \ErrLang$.
\end{itemize}

\noindent\fbox{\textbf{2.}}
If $\run(w) \notin \MRuns_\M$, there is some prefix $w' = w'' \cdot (\instruction,m,t)$ of $w$ such that $\run(w') \notin \MRuns_\M$ and $\run(w'') \in \MRuns_\M$.
This is because at least all prefixes of the form $(\pool,m,t)^\ast$ are encoding the empty run, which belongs to $\MRuns_\M$.

We look at the different cases for $\instruction$ to show that $w' \notin \NotErrLang$.
Again, we let $f = \fracpart{t} \in \FracPart$.
\begin{itemize}
\item Suppose $\instruction \in \Instructions_\inc^c$.
We have $\run(w') \notin \MRuns_\M$, $\run(w'') \in \MRuns_\M$, and $\run(w') = \run(w'') \xrightarrow{\instruction} (\state,\wordToVal(w'))$.
This means that $\wordToVal(w')$ is not the correct valuation, that is $\wordToVal(w') \neq \wordToVal(w'')[c\increment]$.
From this, we deduce that $f$ is active for $c$ in $w''$, otherwise $\wordToVal(w')$ would be correct.
Therefore $w' \in \ErrLangOne$, which implies $w' \notin \NotErrLang$.

\item Suppose $\instruction \in \Instructions_\dec^c$.
There are two possible reasons for $\run(w')$ not to be in $\MRuns_\M$: either $\wordToVal(w')$ is wrong, or $\wordToVal(w'')(c) = 0$.
In the first case, we deduce similarly to the increment case that $f$ is inactive for $c$ in $w''$.
In the second case, it means that there are no active fractional part for $c$ in $w''$, therefore $f$ can only be inactive also.
In both cases, $f$ being inactive does not mean that $w' \in Err_2$, as $f$ might be missing from the pool.
However, this still means that $w' \notin \NotErrLangTwo$ because $\NotErrLangTwo$ only accepts if the last fractional part is active.
Then we have that $w' \notin \NotErrLang$.

\item If $\instruction \in \Instructions_\zt^c$ then we immediately get a contradiction because by definition $\wordToVal(w')$ will be correct so $\run(w')$ must be in $\MRuns_\M$.\qedhere
\end{itemize}
\label{apxproof:correct encoding game}
\end{proof}


Note that we cannot state $w' \in \NotErrLang$ for the first part and $w' \in \ErrLang$ for the second part, again because the pool may be missing the crucial fractional part needed for this.
However, this lemma is enough to prove correctness of the reduction:

\begin{theorem}
\label{thm:bounded}
$\M$ is bounded if and only if \Monitor has a winning controller.  
\end{theorem}

\begin{proof}
\fbox{$\Rightarrow$}
Assume $\M$ is bounded by $k$.
We build a controller $\Controller$ for \Monitor that will correctly detect the first time a mistake happens in the encoding of the run.
This assumes that before that point the number of active fractional parts for a given counter never goes above $k$.
This controller is a 1-resetting \DTA using $4k$ clocks over input timed words with alphabet $\TimerAlphabet$ and outputs in $\MonitorAlphabet$.

Let $\Cl = \{x_i^c \mid 1 \leq i \leq k, c \in \Counter\}$ be the set of clocks of $\Controller$.
Intuitively, clocks $x_1^c, \ldots, x_k^c$ are assigned to counter $c$ and will ``store'' fractional parts that are active for $c$.
By ``storing'' a fractional part $f$, we mean that this clock has value 0 exactly at timestamps whose fractional part is $f$.
The state space of $\Controller$ keeps track of which clocks are considered active for their respective counters.
The initial state is the empty set.
$\Controller$ also ignores the initial pool of data and always outputs $\ok$ on those.

When reading an increment for $c \in \Counter$ with time $t$, $\Controller$ outputs $\err$ if any of the clocks indicated as active for $c$ by the state is at value exactly 0, and $\ok$ otherwise.
If it has output $\err$, it goes to a sink state that always output $\ok$.
Otherwise, it resets the first clock $x_i^c$ marked as inactive, and then it marks this clock as active.
If there is no such free clock, it goes to an error state that outputs whatever.

Similarly, when reading a decrement for $c$ with time $t$, $\Controller$ outputs $\err$ if no active clock has value 0, $\ok$ otherwise.
By construction, no two clocks can have value 0 at the same time.
If the output was $\err$, it goes to the same sink state as before.
Otherwise, it removes that clock from the set of active clocks in the state.

Reading a zero-test transition for counter $c$ makes $\Controller$ output $\ok$ and reset the set of active clocks for $c$ to $\emptyset$.


First we show that $\Controller$ is \emph{accurate} for runs bounded by $k$: 
\begin{lemma}\label{lemma:controller accurate}
For any play $w$ that conforms to $\Controller$ that never has more than $k$ active fractional parts for any counter, $\run(w) \in \MRuns_\M$ if and only if $\projB(w) \in \ok^\ast$.
Moreover, if $\projB(w) \in \ok^\ast$, the valuation of a counter is exactly the number of clocks marked as active by the state of $\Controller$, each of those clocks stores one of the active fractional parts, and those clocks are pairwise different.
\end{lemma}

\begin{nestedproof}
We show this by induction on $w$.
Again, this is trivial for any play of the form $(\pool, \ok, t)^\ast$.
Assuming the above holds for $w$ and that $\Controller$ has always output $\ok$ so far, let us consider some continuation $w' = w \cdot (\instruction,t)$ with $f = \fracpart{t}$.

Case 1a: If $\instruction$ is a $c$ incrementing instruction and $f$ is inactive for $c$ in $w$, then $\wordToVal(w') = \wordToVal(w)[c\increment]$ and therefore $\run(w') \in \MRuns_\M$.
If $\wordToVal(w')(c) > k$ then the proof is finished and $\Controller$'s behaviour after this point does not matter.
Otherwise, $\wordToVal(w)(c) < k$ and by induction hypothesis there are exactly $\wordToVal(w)(c)$ distinct fractional parts active for $c$ in $w$ and in exactly as many clocks of $\Controller$ assigned to $c$.
Moreover, $f$ is in none of them.
In that case, $\Controller$ outputs $\ok$, and stores $f$ in one free clock assigned to $c$, which we know there is at least one.
Then all properties are satisfied by $w'$.

Case 1b: If $\instruction$ is a $c$ incrementing instruction and $f$ is active for $c$ in $w$, then $\run(w') \notin \MRuns_\M$.
Again by induction hypothesis $f$ must be in one of the clocks of $\Controller$ assigned to $c$ and marked as active.
Therefore, $\Controller$ outputs $\err$ and we are done.

Cases 2a and 2b for a decrementing instruction with an active or inactive fractional part respectively are very similar.

Case 3 for a zero-test instruction is also easy: as zero-test instructions are always available under free-test semantics, $w' \in \MRuns_\M$ is immediate, $\Controller$ always outputs $\ok$, and all clocks are marked as inactive which coincide with the valuation being set to 0.
\end{nestedproof}

Returning to the proof of Theorem \ref{thm:bounded}, 
we now show that $\Controller$ is a winning controller for \Monitor.
Let $w$ be a play.

If $\run(w) \in \MRuns_\M$, then we know that there is at most $k$ active fractional parts for any counter in all prefixes of $w$ by boundedness assumption, and therefore that $\projB(w) \in \ok^\ast$ by Lemma~\ref{lemma:controller accurate}.
Moreover, by Lemma~\ref{lemma:correct encoding game}, all prefixes of $w$ are not in $\ErrLang$.
Clearly, $w$ is not in $\Win$.

If $\run(w) \notin \MRuns_\M$, then by Lemma~\ref{lemma:correct encoding game} there exists some prefix $w' = w'' \cdot (\instruction,m,t)$ such that $\run(w'') \in \MRuns_\M$, $\run(w') \notin \MRuns_\M$, and $w' \notin \NotErrLang$.
Therefore, by Lemma~\ref{lemma:controller accurate}, $\projB(w'') \in \ok^\ast$ and $\Controller$ outputs $m = \err$ on $w'$.
$w''$ is not in $\Win$ for the same reason as before, and since $w' \notin \NotErrLang$, $w' \notin \Win$ either.
Then any continuation, including $w$, has $\projB$ of the form $\ok^\ast \cdot \err \cdot \ok^+$ and therefore is not in $\Win$.

\medskip

\noindent\fbox{$\Leftarrow$}
Let $\run$ be an unbounded run of $\M$.
Assume towards a contradiction that there exists some controller $\Controller$ that is winning for \Monitor.
Let $G \subset \Gu(\Cl)$ be the finite set of guards appearing in $\Controller$, and $k = |G|$.
Given some guard $g \in G$ and a clock valuation $\val$, let $T(\val,g) = \{t \in \Rpos \mid \val+t \models g\}$.
As $g$ is a conjunction of constraints, $T(\val,g)$ is either empty, a singleton, or an interval.
Thus, if there are two distinct $t,t' \in T(\val,g)$ then any $t < t'' < t'$ is also in $T(\val,g)$.

At some point, $\run$ goes from $k+1$ to $k+2$ for the valuation of some counter $c$ using some incrementing instruction $\instruction$.
Let $w$ be a timed word that correctly encodes $\run$ until just before $\instruction$, with fractional parts $f_1,\dots,f_{k+1}$ active at the end of $w$.
For simplicity, we take $f_1 < \dots < f_{k+1}$, and assume that $w$ has those fractional parts in the pool.
Moreover, we also put in the pool an inactive fractional part between every $f_i$ and $f_{i+1}$.

Let us now consider what $\Controller$ does on this word.
Necessarily $\Controller$ must have output only $\ok$ so far, otherwise \Timer would win.
On a new action $(\instruction,t)$, it must necessarily output $\err$ if $\fracpart{t} \in \{f_1,\dots,f_{k+1}\}$ as those are active fractional parts, and therefore should not be used for a new increment.
But on any other timestamp, if the corresponding fractional part is in the pool, it must output $\ok$ to not lose immediately.
From $\Controller$'s current state, there are less than $k$ transitions outputting $\err$.
So at least one of them can be fired with at least two distinct active factional parts.
But as we have seen earlier, this means any timestamp between those two can also fire this transition.
We have at least one inactive fractional part in the pool between the two active ones, and there is a timestamp with this fractional part that can fire the transition, outputting $\err$.
Thus, we get a contradiction.
\label{apxproof:bounded}
\end{proof}

Note that the existence of a controller and a strategy are not equivalent here: 
\Monitor always has a winning strategy, but only has a winning controller if $\M$ is bounded.


\section{Conclusion}\label{sec:conc}

The main contribution of this paper is to solve the problem left open 
in~\cite{Clemente0P-20,Piorkowski-PhD_22} by proving that all the variants of timed synthesis problems are undecidable
for winning sets specified already by $\NTA_1$.
The only exception is \cref{thm:Radek} in
\cref{sec:boundedness} where we need to extend
(reachability) $\NTA_1$ by a very limited form of $\varepsilon$-transitions.
While we believe that these $\varepsilon$-transitions may be eliminated, we consider the
current result as a satisfactory solution of the open problem:
all our undecidability results translate from timed setting to data setting
where winning sets are specified using nondeterministic one-register automata ($\NRA_1$),
and neither $\varepsilon$-transitions nor guessing are needed there.

One of the motivations for studying timed/data synthesis in~\cite{Clemente0P-20,Piorkowski-PhD_22}
was a potential application to solving the \emph{deterministic separability} question:
given two nondeterministic \NTA (resp.~\NRA) over finite timed (resp.~data)
words, with disjoint languages $L_1, L_2$,
is there a \DTA (\DRA) whose language \emph{separates} $L_1$ from $L_2$,
namely includes one of them and is disjoint from the other.
Decidability of resource-bounded timed/data synthesis is used 
in~\cite{Clemente0P-20,Piorkowski-PhD_22} to obtain decidability of resource-bounded
deterministic separability, where one a priori bounds the number of clocks/registers
in a separating automaton.
Decidability status of the unrestricted deterministic separability still remains open.
Also a related problem of deterministic membership, where given 
a nondeterministic \NTA, one asks if its language is accepted by some 
deterministic \DTA, is undecidable even for $\NTA_1$ \cite{Clemente20}.
On the other hand it decomes decidable if the number of clocks in a deterministic 
timed automaton is a priori bounded \cite{Clemente20}.

Knowing that $\NTA_1$ winning sets yield ubiquitous undecidability of synthesis problems,
a natural follow-up questions is to ask if the situation changes when one restricts to 
subclasses of $\NTA_1$, for instance to \emph{reachability} $\NTA_1$
winning sets.
Since deterministic separability over finite words reduces to timed games with
reachability $\NTA$ winning conditions, this could help solve deterministic separability
of $\NTA_1$ languages.

Finally, we recall that we do not know if the games studied in this paper are determined,
namely if one of the players has always a winning strategy.

\newpage


\bibliography{bib.bib}

\end{document}